\documentclass[5p,times,twocolumn]{elsarticle}

\usepackage[T1]{fontenc}
\usepackage{amsmath,amssymb}
\usepackage{algorithm}
\usepackage{algpseudocode}
\usepackage{graphicx}
\usepackage{booktabs}
\usepackage{xurl}
\usepackage[hidelinks,hypertexnames=false]{hyperref}

\newtheorem{theorem}{Theorem}
\newtheorem{proposition}{Proposition}
\newproof{proof}{Proof}

\journal{}
\date{12 October 2026}
\hypersetup{
  pdftitle={Hypergeometric Bridges for Multidimensional Grid Visibility},
  pdfauthor={Ibrahim Ibrahim and Wilm Decre}
}
\makeatletter
\def\ps@pprintTitle{%
  \let\@oddhead\@empty
  \let\@evenhead\@empty
  \def\@oddfoot{\hfil\thepage\hfil}%
  \let\@evenfoot\@oddfoot}
\makeatother
\biboptions{numbers,sort&compress}

\begin{document}
\begin{frontmatter}

\title{Hypergeometric Bridges for Multidimensional Grid Visibility}

\author[aff1]{Ibrahim Ibrahim\corref{cor1}\fnref{orcid1}}
\ead{ibrahim.ibrahim@kuleuven.be}
\author[aff1]{Wilm Decré\fnref{orcid2}}

\cortext[cor1]{Corresponding author}
\fntext[orcid1]{ORCID: 0000-0001-6840-558X}
\fntext[orcid2]{ORCID: 0000-0002-9724-8103}

\affiliation[aff1]{
	organization={Robotics, Automation and Mechatronics (RAM), Department of Mechanical Engineering,
		KU Leuven, and Flanders Make@KU Leuven},
	addressline={Celestijnenlaan 300, box 2420},
	city={Leuven},
	postcode={B-3001},
	country={Belgium}}

\begin{abstract}
	Established interpolation-based methods efficiently compute visibility from
	one source to every grid cell using values computed closer to the source.
	Interpolation yields fractional values even on binary grids of free or fully
	blocking cells. Two-dimensional path-counting theory interprets
	these values as fractions of unobstructed grid paths, but three-dimensional methods
	use different predecessor structures. Interpreting these values consistently
	across dimensions and on nonbinary grids remains unresolved.

	We provide a unified mathematical account that gives every
	visibility value a precise meaning in any dimension. We establish when grid
	refinement yields eventual agreement with continuous visibility, implement
	parallel solvers, and compare four three-dimensional methods in accuracy and
	throughput.

	We define cell transmission: zero blocks visibility, one preserves it, and
	intermediate values attenuate it proportionally. We prove each computed value
	equals the uniform \emph{finite-path average} of transmission products over
	finitely many source-to-cell grid paths. It is the
	fraction of clear paths on binary grids and the mean path transmission on
	nonbinary grids. For a source--target query, the
	fractions of paths passing through cells at each step form a
	hypergeometric distribution in 2D, which we call the \emph{hypergeometric
	bridge}. On the last step into the target, these probabilities equal the
	interpolation weights. Higher-dimensional interpolation weights arise from
	products of hypergeometric bridges.
	For any fixed dimension, the full-field algorithm remains linear in grid size.
	Exact path enumeration reproduces every checked visibility value. NumPy,
	OpenMP, and CUDA agree numerically, and our OpenMP solver achieves the highest
	throughput in the four-method CPU comparison at every tested resolution.
\end{abstract}


\end{frontmatter}

\section{Introduction}
\label{sec:introduction}

At its simplest, continuous visibility asks whether an obstacle blocks the
straight line segment from a source point to a target point.  On a grid, many
applications need that answer at every target rather than at only one: a robot may evaluate
visibility at every candidate location, a sensor study may compare coverage
throughout a room, and an architect may examine visibility throughout a voxel
model of a building.  Casting an independent ray to every target is direct, but
the ray-by-ray calculation solves each query in isolation, whereas the
full-field approach studied here assigns one scalar visibility value to every
cell and computes all cell values together.
Figure~\ref{fig:visibility-3d-showcase} previews the computed full-volume
visibility field in two three-dimensional scenes.
Such a field should do more than save work: every propagated value should have
a precise meaning.

\begin{figure*}[tp]
	\centering
	\begin{minipage}[t]{0.48\textwidth}
		\centering
		\includegraphics[width=\linewidth]{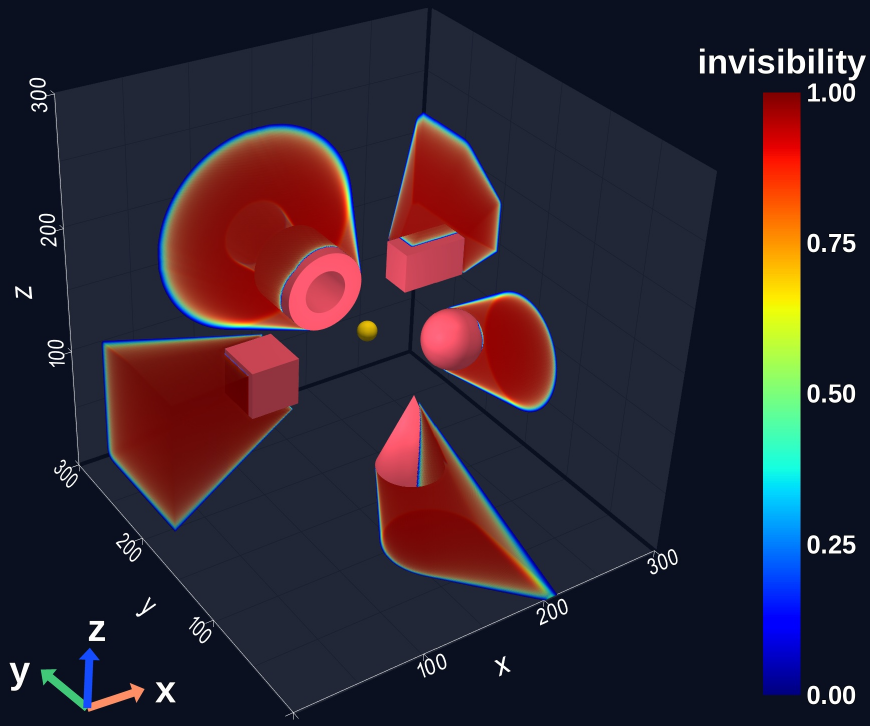}\\[-1mm]
		\small (a) Scene A: volumetric view
	\end{minipage}\hfill
	\begin{minipage}[t]{0.48\textwidth}
		\centering
		\includegraphics[width=\linewidth]{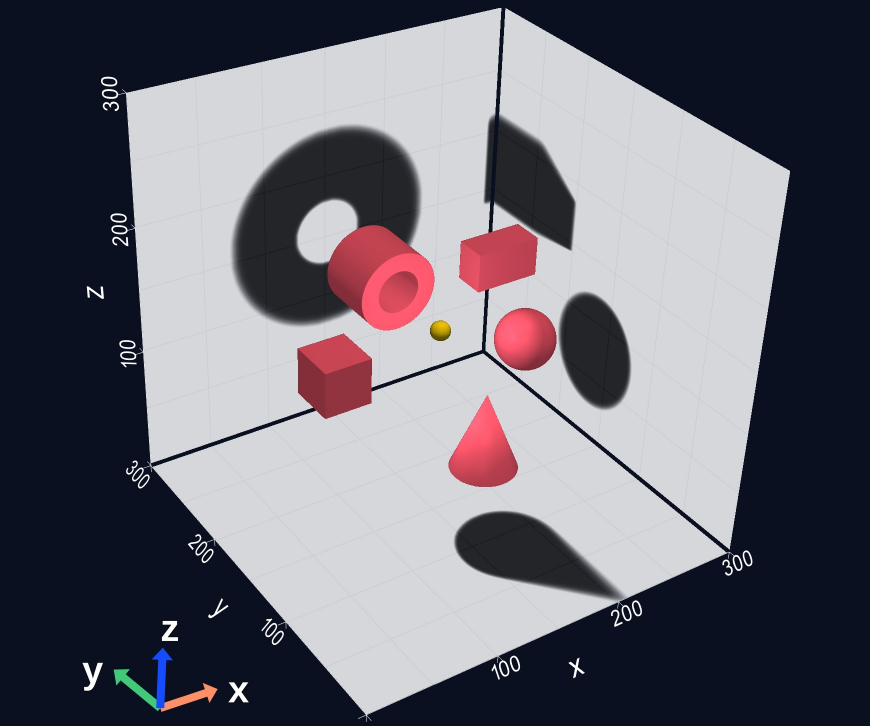}\\[-1mm]
		\small (b) Scene A: plane views
	\end{minipage}
	\vspace{1.5mm}

	\begin{minipage}[t]{0.48\textwidth}
		\centering
		\includegraphics[width=\linewidth]{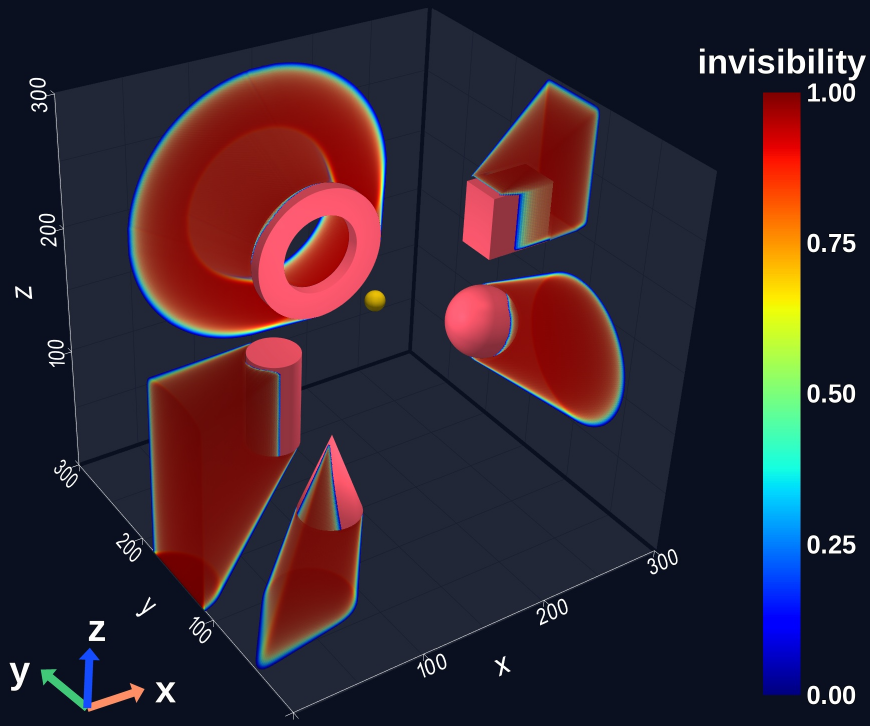}\\[-1mm]
		\small (c) Scene B: volumetric view
	\end{minipage}\hfill
	\begin{minipage}[t]{0.48\textwidth}
		\centering
		\includegraphics[width=\linewidth]{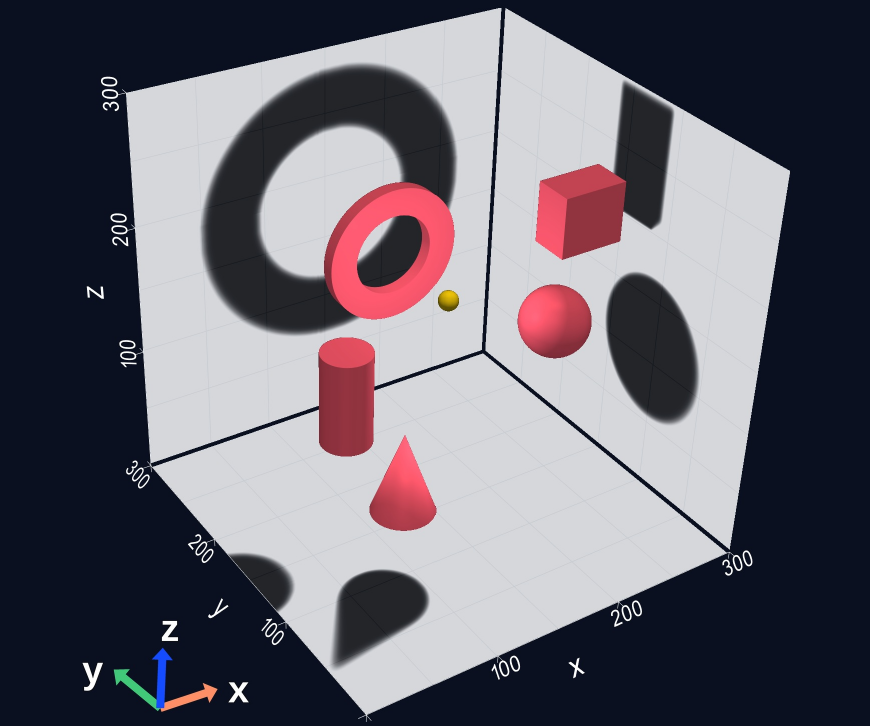}\\[-1mm]
		\small (d) Scene B: plane views
	\end{minipage}
	\caption{Illustrative outputs of the three-dimensional full-field computation
		studied here.  The yellow sphere marks the source, and the pink surfaces mark
		opaque obstacles.  The left column shows cells below an illustrative
		visibility threshold throughout each volume.  The right column shows the
		corresponding field on the floor, rear boundary, and right boundary.  For
		display only, the colour scale rescales values below the threshold so that
		zero denotes the threshold and one denotes a computed visibility value of
		zero.}
	\label{fig:visibility-3d-showcase}
\end{figure*}

A causal grid sweep constructs this field from the source outward.  Cells are
grouped into layers by their largest absolute coordinate displacement from the
source, with a coordinate attaining that maximum designated as dominant for the
cell.  Under this dominant-axis ordering, processing the layers outward ensures
that every cell value on the immediately preceding layer is available when the
sweep reaches a target.  For each target, the interpolation point is where the
ray from the source through that target crosses the immediately preceding
layer.  The intersection need not coincide with a cell centre, so linear
interpolation in two dimensions (2D), or bilinear interpolation in three
dimensions (3D),
combines nearby, already computed cell-centre values into the field value
carried to the target.

The target then applies its prescribed cell transmission as a scalar
multiplier: zero blocks the carried value, one preserves it, and any
transmission between zero and one attenuates it.  Together, the interpolation
step and the target-cell multiplication define the \emph{dominant-axis
	visibility recurrence}: beginning from an initial visibility value
assigned to the source cell, this recursive rule expresses each non-source
visibility value in terms of values on the preceding layer, while the causal
sweep evaluates those dependencies layer by layer from the source outward.
Even when every cell
transmission is binary---either zero or one---interpolation can produce a field
value strictly between those limits, whereas the original straight-line
question has only a yes-or-no answer.  The resulting scalar field is therefore
not self-interpreting: a value between zero and one does not by itself specify
either the fraction of the source--target segment that is clear or the
probability that the entire segment is clear.  The mathematical interpretation
assigned to that fractional value determines what can be claimed about
exactness, thresholds, and agreement with continuous visibility.  Interpreting a
fractional field value begins with a basic question: what exactly does the
dominant-axis interpolation average?

Several parts of the answer are established.  Tsai et al.\ formulated
visibility through a source-centred transport equation and showed how the
equation's radial ordering supports one-pass interpolation in two and three
dimensions
\cite{Tsai2004Visibility}.  Kao and Tsai later developed the dynamic-programming
and Hamilton--Jacobi analysis of source-centred radial visibility
\cite{KaoTsai2008Visibility}.  Together, the source-centred transport equation
and the dynamic-programming and Hamilton--Jacobi analyses establish why a
one-pass grid sweep can proceed causally from the source and how each local
interpolation is constructed.  Goldstein et al.\ later gave the 2D scalar
visibility value for a binary map a direct grid-path interpretation: in
SpaceAnalysis it is the fraction of shortest grid paths that remain
unobstructed as they advance one layer outward from the source at each step.
Their later
path-counting analysis proved that the dominant-axis visibility recurrence they
analysed computes the same fraction of unobstructed paths.  Goldstein et al.\ also expressed the
number of paths through an intermediate cell as the product of the number of
shortest grid paths from the source to the cell and the number of shortest
continuations from the cell to the target
\cite{Goldstein2020SpaceAnalysis,Goldstein2022PathCounting}.

Fisher-Gewirtzman et al.\ used voxel-based volumetric visibility for
map-wide 3D analysis in urban environments
\cite{FisherGewirtzman2013VoxelVisibility}.  Goldstein et al.\ later
incorporated 3D visibility into the broader Voxel-based Architectural Space
Analysis (VASA) system \cite{Goldstein2023VASA}.  Both
Fisher-Gewirtzman et al.'s urban visibility analysis and VASA treat a complete
visibility volume, rather than one source--target query, as the object of
analysis.

Tsai et al.'s source-centred transport equation and its causal radial ordering,
Goldstein et al.'s binary 2D interpretation in terms of unobstructed grid paths,
and the full-volume 3D
analyses of Fisher-Gewirtzman et al.\ and VASA establish essential pieces of
the answer.  What remains is one account that connects the source-centred radial
transport derivation and the established 2D path counts to an
endpoint-conditioned discrete distribution over the cells on each intermediate
layer.  For each cell on that layer, the distribution gives the probability
that a uniformly selected path passes through that cell.  Such an account would
keep the field interpretable when the recurrence extends to 3D, cell transmissions
become nonbinary, and computation scales to full volumes.  Without that
connection, ``exact,'' ``probabilistic,'' and even ``visibility'' can refer to
different mathematical objects.

This paper supplies that connection by identifying the grid paths averaged by
the recurrence and deriving the probability law governing how those paths pass
through each intermediate layer.  To interpret a stored visibility value,
treat an arbitrary non-source grid cell as the target.  We define an explicit
finite family of monotone grid paths
from the source to that target, which we call \emph{canonical paths}.  Each
canonical path advances without reversing any coordinate and contributes the
product of the prescribed cell transmissions it encounters.  We prove that
the scalar visibility value computed at the target by the dominant-axis
visibility recurrence is the \emph{finite-path average}: the value assigned to
the source cell multiplied by the uniform average of those path-transmission
products.  The finite family of canonical paths and its finite-path average
together form the \emph{finite-path model}.
Since the target is arbitrary, every non-source grid cell has its own
source-to-cell canonical path family, so the same interpretation applies
to every scalar visibility value in the full field.  When the visibility field
is initialized to one at the source cell, the average is the fraction
of canonical paths that remain clear for binary transmissions and the mean of
their transmission products for deterministic nonbinary transmissions.  The
causal sweep evaluates the recurrence as a dynamic program, computing the
average for every target without enumerating paths and requiring constant work
per grid cell in 2D and 3D.

To study a different quantity, now keep one source and one target fixed and
ask how the corresponding source-to-target canonical path family is
distributed before reaching the target.  On each 2D grid layer between the
endpoints, every path
passes through exactly one cell.  We call the probability that a uniformly
selected path passes through a particular cell the cell's
\emph{path-incidence probability}.  Unlike the scalar visibility value stored
at that cell, which averages transmission products over paths ending there,
the path-incidence probability concerns paths continuing to the fixed target.
On a given layer, the number of transverse, or nondominant, steps already
taken identifies the cell reached.  Because the fixed endpoints prescribe the
total number of transverse steps, that count follows a discrete
sampling-without-replacement distribution conditioned on both endpoints---the
\emph{hypergeometric bridge}.  The bridge thus describes how the fixed path
family is distributed across each layer.  These path-incidence probabilities
arise from uniform path selection, not map uncertainty about whether cells are
free or occupied.

In 3D, every ordered path step advances the dominant coordinate, and the
\emph{transverse-advance schedule} of each transverse coordinate is the subset
of those steps on which that coordinate also advances.
Selecting the two schedules independently and uniformly produces a uniformly
selected path from the 3D family.  Because the endpoint fixes the total number
of advances in each schedule, the cumulative advance count for each coordinate
follows a one-dimensional hypergeometric bridge, and the joint distribution of
the two counts is the product of those bridges---the \emph{3D product bridge}.
At the immediately preceding layer, the 3D product bridge gives the
probabilities of four target-entry cases, according to whether neither, either,
or both transverse coordinates advance on the final path step into the target.
These four target-entry probabilities are exactly the coefficients with which
the 3D bilinear interpolation combines the four surrounding cell-centre values
on the preceding grid face.  Because the dominant-axis recurrence equals the
uniform average of the transmission products contributed by the canonical
paths, evaluating every canonical path directly provides an independent check
of the recurrence.  Exact agreement between the resulting path average and the
recurrence remains distinct from agreement with continuous visibility.

Together, four contributions establish a theory--algorithm--\hspace{0pt}implementation
chain for the established dominant-axis recurrence: a source-relative
formulation, the finite-path model, explicit scalar and probability semantics,
and reproducible full-volume computation.

\begin{itemize}
	\item We derive the 2D and 3D visibility recurrence in source-relative
	      coordinates for every grid source position and source--target direction,
	      and prove that the resulting outward sweep is causal and that every
	      computed value remains in $[0,1]$.
	\item We prove that the recurrence computes the finite-path average over the
	      corresponding canonical path family.  We identify the 2D intermediate-cell
	      path-incidence law as the hypergeometric bridge and derive the 3D product
	      bridge whose target-entry probabilities equal the four bilinear weights.
	      We derive the multilinear recurrence and prove the path-average identity in
	      any fixed spatial dimension, with work linear in the number of returned grid
	      values.
	\item We distinguish the path-incidence probability induced by uniform selection
	      within the canonical path family from map uncertainty, and use the
	      finite-path average to interpret scalar visibility under deterministic binary
	      and nonbinary cell transmissions.  When the visibility field is initialized
	      to one at the source cell, classifying a binary-map target as visible at a
	      chosen threshold requires at least that fraction of the target's canonical
	      paths to remain clear.  We also establish the grid-refinement conditions
	      under which these thresholded classifications agree with continuous
	      visibility on every sufficiently fine grid for fixed queries whose segments
	      have positive obstacle clearance or enter an obstacle interior
	      (Theorem~\ref{thm:nongrazing-refinement-consistency}).
	\item We provide open-source GitHub code comprising exact-arithmetic evaluators
	      of the recurrence in 2D, 3D, and any fixed spatial dimension and
	      full-volume NumPy, OpenMP, and CUDA implementations of the four-parent
	      3D recurrence.  Exact
	      enumeration on tractable path families checks that the recurrence returns
	      the finite-path average, while numerical agreement checks that the three
	      full-volume implementations return the same visibility field up to the
	      reported floating-point differences.  We compare thresholded classifications
	      with continuous visibility under grid refinement and CPU throughput with
	      three 3D alternatives.  In the CPU comparison on the benchmark laptop, the
	      four-parent implementation achieves the highest throughput among all four
	      methods at every tested resolution.
\end{itemize}

This article extends our earlier conference paper
\cite{IbrahimEtAl2024Visibility}, which presented the two-parent 2D sweep and
applied it to heuristic path planning.  Relative to that work, the journal
article supplies the corrected source-relative derivation and complete tie
handling, establishes the sweep's finite-path and probabilistic interpretation,
and extends the model to 3D and any fixed spatial dimension.  It also states the
scalar-transmission limits and the conditions for agreement with continuous
visibility.  Beyond these theoretical extensions, the journal article adds
exact-arithmetic verification, full-volume CPU and GPU implementations, and
evaluation against three 3D alternatives.

Section~\ref{sec:related-work} positions this paper's four-parent dominant-axis
recurrence against prior work, and Section~\ref{sec:recurrence} derives the
source-relative recurrence.
Sections~\ref{sec:exact-2d-path-model}--\ref{sec:multidimensional-extension}
develop the finite-path model and its transmission semantics from 2D through
any fixed spatial dimension, derive the bridge laws, and establish eventual
agreement with continuous visibility under grid refinement.
Sections~\ref{sec:implementations}
and~\ref{sec:validation} present the full-volume algorithm and implementations,
verify that they preserve the recurrence, and then evaluate classification
accuracy under grid refinement and computational performance across all four
methods.
Section~\ref{sec:conclusion} closes with the model's scope, limitations, and
conclusions.

\section{Related Work}
\label{sec:related-work}

The word \emph{visibility} can name different questions.  An independent ray
test asks whether one source--target segment intersects an obstacle.  A
map-wide sweep instead propagates a scalar over the whole grid, and that scalar
may encode binary membership in a geometrically visible set, the fraction of
unobstructed paths in a discrete path family, or accumulated transmission.
Although ray tests and
map-wide sweeps can share a radial ordering, they need not compute
interchangeable quantities.  We therefore organize the closest work by the
quantity each method computes, rather than by the common label
\emph{visibility}.

Tsai et al.\ formulated visibility in a level-set framework, transported a
ray-wise quantity with a source-centred partial differential equation (PDE), and
derived ordered 2D and 3D interpolation schemes
\cite{Tsai2004Visibility}.  Kao and Tsai subsequently developed the
dynamic-programming and Hamilton--Jacobi analysis, including stronger
convergence results \cite{KaoTsai2008Visibility}.  Oberman and Salvador cast
visibility as a local obstacle problem whose solution represents a
star-shaped visible region and proved convergence for a monotone scheme
\cite{ObermanSalvador2020Visibility}.  These studies establish radial PDE
formulations and causal sweeps for map-wide visibility, but they propagate
ray-wise geometric quantities or star-shaped envelopes rather than the
multiplicative cell-transmission field considered here.

Goldstein et al.\ connected dominant-axis interpolation to shortest-grid-path
counting in 2D \cite{Goldstein2020SpaceAnalysis,Goldstein2022PathCounting}.
For a binary map, the visibility recurrence they analyse returns the fraction of
unobstructed shortest grid paths.  Their product of forward and backward path counts,
divided by the total number of source--target paths, is algebraically identical
to the probability that a uniformly selected path passes through an
intermediate cell.
We therefore do not claim the 2D coefficients, the recurrence that computes
the fraction of unobstructed paths,
or through-cell counting identity as new.  Building on these identities, this
paper identifies and develops the hypergeometric bridge for the 2D path
family---the endpoint-conditioned discrete distribution that assigns to each
cell on an intermediate layer the probability that a uniformly selected path
passes through that cell---and extends the finite-path average to nonbinary cell
transmissions.

Volumetric visibility has also been computed directly on voxel grids.
Fisher-Gewirtzman et al.\ developed a voxel-based method for volumetric
visibility analysis in 3D urban environments that processes each voxel once
and returns continuous values between fully visible and fully invisible
\cite{FisherGewirtzman2013VoxelVisibility}.

Goldstein et al.\ describe two linear 3D families
\cite{Goldstein2022PathCounting}.  The standard six-neighbour recurrence
combines up to three axial neighbours, each one grid step toward the source
along a coordinate axis.  Their 26-neighbour method
divides directions into 48 tetrahedral regions and combines up to three
predecessors connected by axis-aligned, face-diagonal, or space-diagonal moves.
In both published recurrences, a predecessor contributes only when the local
predecessor--target edge passes a binary visibility test.

The later VASA system provides a mature family of architectural voxel analyses
and states that point visibility uses the standard six-neighbour family
\cite{Goldstein2023VASA}.  The downloadable VASA distribution contains compiled
libraries rather than source code, and its visibility operations apply a
user-specified threshold before returning a voxel model rather than the
underlying scalar field.  The six-neighbour comparator used later in this
paper therefore represents the published formula family, not a reproduction of
VASA's packaged implementation or output.  The four-parent dominant-axis
recurrence formalized in this paper instead interpolates the four corners of the
grid face crossed one layer inward along the dominant coordinate, applies
destination-cell transmission, and retains the resulting scalar field.  We develop the
corresponding 3D product bridge to explain the four bilinear coefficients of
the four-parent recurrence.  The six-neighbour, tetrahedral, and
four-parent bilinear recurrences average over distinct discrete path families
determined by their predecessor rules, so their scalar values are not
interchangeable.

Radial transport, causal sweeps, 2D interpolation and path counts, and several
3D voxel families precede this work.  We connect the source-relative
dominant-axis recurrences to the canonical path families whose transmission
products they average, the hypergeometric and 3D product bridges,
cell-transmission semantics, and
validated NumPy, OpenMP, and CUDA implementations of the four-parent 3D
recurrence.  All three implementations compute full volumes while preserving
the finite-path model of the dominant-axis recurrence, so the value returned at
each grid cell remains the average of the transmission products over that
cell's source-to-cell canonical path family.  For test cases with
tractable path counts, exhaustive enumeration independently checks that the
computed fields agree with the finite-path averages.
Section~\ref{sec:recurrence} begins by defining the propagated values and
weights.

\section{Dominant-Axis Scalar Visibility Recurrence}
\label{sec:recurrence}

\subsection{Input, output, and source-relative transport}

Let $\mathcal{G}_h\subset\mathbb{R}^d$, with spatial dimension
$d\in\{2,3\}$, be a rectangular Cartesian grid of cell centres with spacing
$h$, and let the source $\boldsymbol{s}\in\mathcal{G}_h$ be one of those
centres.  Each grid location
$\boldsymbol{q}$ is assigned a cell transmission
$\tau(\boldsymbol{q})\in[0,1]$: zero is fully blocking, one is free, and values
between zero and one attenuate the propagated scalar.  The
output $V(\boldsymbol{q})\in[0,1]$ is a scalar visibility value.  The initial
visibility value $V_s=V(\boldsymbol{s})$, assigned at the source cell, supplies
the boundary condition for the outward recurrence.  Unless stated otherwise,
we set $V_s=1$, which we call the
\emph{unit-source convention}.  To include source-cell transmission, the
alternative \emph{transmitting-source convention} sets
$V_s=\tau(\boldsymbol{s})$.  For a
fixed $\tau$, the scalar field $V$ is deterministic.  This scalar transmission
model does not represent scattering, refraction, or multiple-bounce light
transport.

If $\tau$ is instead interpreted as a field of marginal free probabilities,
with $\tau(\boldsymbol{q})$ denoting the probability that cell
$\boldsymbol{q}$ is free when considered individually, then the product of the
cellwise probabilities along a path represents the path's expected
transmission only when the required joint moments factorize.  Equivalently, the
probability that all cells on a path are free must equal the product of their
individual free probabilities, as it does when their states are independent.
The exact 2D path model will distinguish deterministic cell transmission from
map uncertainty.

First ignore attenuation and ask how a value is carried along a ray from the
source, following the source-centred radial-transport formulations in
\cite{KaoTsai2008Visibility,Tsai2004Visibility}.  Let
$\boldsymbol{x}\in\mathbb{R}^d\setminus\{\boldsymbol{s}\}$ be a spatial point,
not necessarily a cell centre in $\mathcal{G}_h$.  Define its source-relative
displacement, the vector from $\boldsymbol{s}$ to $\boldsymbol{x}$, by
\begin{equation}
	\boldsymbol{\rho}=\boldsymbol{x}-\boldsymbol{s}.
	\label{eq:source-relative-displacement}
\end{equation}
Fix a nonzero direction vector $\boldsymbol{\rho}_0$.  The corresponding ray
from the source is $\boldsymbol{x}(\lambda)=\boldsymbol{s}+
	\lambda\boldsymbol{\rho}_0$, where $\lambda>0$ is a scalar ray parameter.  For
a scalar field $u$, transport without attenuation means that $u$ does not
change as $\lambda$ moves a point outward along this fixed source ray.  Since
$\boldsymbol{x}'(\lambda)=\boldsymbol{\rho}_0$, the chain rule gives
\[
	\frac{\mathrm{d}}{\mathrm{d}\lambda}u(\boldsymbol{x}(\lambda))
	=\boldsymbol{\rho}_0\mathbin{\cdot}\nabla u(\boldsymbol{x}(\lambda))=0.
\]
Thus $u$ may vary between source rays but remains constant along each one.
Since $\boldsymbol{\rho}=\lambda\boldsymbol{\rho}_0$ on that ray, the
canonical source-relative transport equation is
\begin{equation}
	(\boldsymbol{x}-\boldsymbol{s})\mathbin{\cdot}\nabla u(\boldsymbol{x})=0,
	\qquad \boldsymbol{x}\ne\boldsymbol{s}.
	\label{eq:radial-transport}
\end{equation}
The source is a prescribed boundary value because the radial direction is not
defined there.  Dividing~\eqref{eq:radial-transport} by
$\|\boldsymbol{x}-\boldsymbol{s}\|$ rescales the equation without changing its
characteristic rays.

\subsection{Following a ray one layer inward}

Consider a target $\boldsymbol{q}\in\mathcal{G}_h$ distinct from the source.
Let $\boldsymbol{n}=(n_1,\ldots,n_d)
	\in\mathbb{Z}^d\setminus\{\boldsymbol{0}\}$ denote its source-relative
displacement measured in grid steps, so that
\begin{equation}
	\boldsymbol{q}=\boldsymbol{s}+h\boldsymbol{n}.
	\label{eq:grid-displacement}
\end{equation}
Let $M=\|\boldsymbol{n}\|_\infty$.  Thus $M$ counts outward layers in the
coordinate with the largest source-relative displacement.  Choose one such
coordinate $m$, called the \emph{dominant coordinate}, with ties assigned to
the smallest coordinate index ($x$, then $y$, then $z$ in 3D).  Define
\begin{equation}
	\sigma_j=\operatorname{sign}(n_j)\quad(j=1,\ldots,d),
	\qquad p_j=\frac{|n_j|}{M}\quad(j\ne m).
	\label{eq:signs-and-ratios}
\end{equation}
The signs encode the source-to-target direction in each coordinate.  Thus,
$+\sigma_j\boldsymbol{e}_j$ points toward the target, whereas
$-\sigma_j\boldsymbol{e}_j$ points back toward the source.  Because $|n_m|=M$,
the corresponding dominant-coordinate ratio would be $p_m=1$:
moving one layer inward changes coordinate $m$ by exactly one grid spacing, so
interpolation is required only in the nondominant coordinates $j\ne m$, where
the ratios $p_j\in[0,1]$ determine the weights.

Backtracking the source--target ray until coordinate $m$ has moved inward by one
grid spacing corresponds to $\lambda=(M-1)/M$ and reaches
\begin{equation}
	\boldsymbol{q}^{\star}
	=\boldsymbol{s}+\frac{M-1}{M}(\boldsymbol{q}-\boldsymbol{s}).
	\label{eq:characteristic-foot}
\end{equation}
In 2D, $\boldsymbol{q}^{\star}$ lies on the grid edge one layer inward along
the dominant coordinate, whereas in 3D it lies on the corresponding grid
face.  Among the grid corners surrounding $\boldsymbol{q}^{\star}$, every corner with nonzero
interpolation weight lies on layer $M-1$ and has therefore already been
computed in the outward sweep.  We call each such contributing corner a
\emph{predecessor}.

Let $\boldsymbol{e}_j$ denote the $j$th coordinate unit vector.  Each corner of
the grid edge or face containing $\boldsymbol{q}^{\star}$ includes the
dominant-coordinate displacement
$-h\sigma_m\boldsymbol{e}_m$, which moves one grid layer toward the source.  To
enumerate the remaining coordinate choices, assign one binary indicator
$\epsilon_j$ to each nondominant coordinate $j\ne m$: setting $\epsilon_j=0$
leaves coordinate $j$ equal to $q_j$, whereas $\epsilon_j=1$ applies
$-h\sigma_j\boldsymbol{e}_j$.  Each binary vector
$\boldsymbol{\epsilon}=(\epsilon_j)_{j\ne m}\in\{0,1\}^{d-1}$ identifies the
grid point
\begin{equation}
	\boldsymbol{q}_{\boldsymbol{\epsilon}}
	=\boldsymbol{q}-h\sigma_m\boldsymbol{e}_m
	-h\sum_{j\ne m}\epsilon_j\sigma_j\boldsymbol{e}_j.
	\label{eq:predecessor-corners}
\end{equation}
The $2^{d-1}$ labels cover two endpoint choices in 2D and four face-corner
choices in 3D.  Each choice receives an interpolation weight below, and any
zero-weight choice drops out of the interpolated value.

For each binary label, define the one-dimensional factors and their
tensor-product weight by
\begin{equation}
	\omega_j(0)=1-p_j,
	\qquad \omega_j(1)=p_j,
	\qquad
	w_{\boldsymbol{\epsilon}}=\prod_{j\ne m}\omega_j(\epsilon_j).
	\label{eq:predecessor-weights}
\end{equation}
The grid points $\boldsymbol{q}_{\boldsymbol{\epsilon}}$, together with their
corresponding weights $w_{\boldsymbol{\epsilon}}$, form the local free-transport
stencil at $\boldsymbol{q}$.
Linear or bilinear interpolation at the ray--layer intersection is then
\begin{equation}
	\mathcal{I}_h V(\boldsymbol{q})
	=\sum_{\boldsymbol{\epsilon}\in\{0,1\}^{d-1}}
	w_{\boldsymbol{\epsilon}}
	V(\boldsymbol{q}_{\boldsymbol{\epsilon}}).
	\label{eq:multilinear-interpolant}
\end{equation}
In 2D, suppose $x$ is dominant and set $p_y=|n_y|/|n_x|$.  With
\[
	\boldsymbol{q}_a=\boldsymbol{q}-h\sigma_x\boldsymbol{e}_x
	-ah\sigma_y\boldsymbol{e}_y,
	\qquad a\in\{0,1\},
\]
the two-parent free-transport interpolant is
\begin{equation}
	\mathcal{I}_hV(\boldsymbol{q})
	=(1-p_y)V(\boldsymbol{q}_0)+p_yV(\boldsymbol{q}_1).
	\label{eq:2d-interpolant}
\end{equation}
The $y$-dominant formula follows by exchanging $x$ and $y$.

Figure~\ref{fig:stencil-geometry-2d} makes the local geometry explicit for the
$x$-dominant target $\boldsymbol{q}=(3,1)$ on a unit grid with source
$\boldsymbol{s}=(0,0)$.

\begin{figure*}[t]
	\centering
	\includegraphics[width=0.82\textwidth]{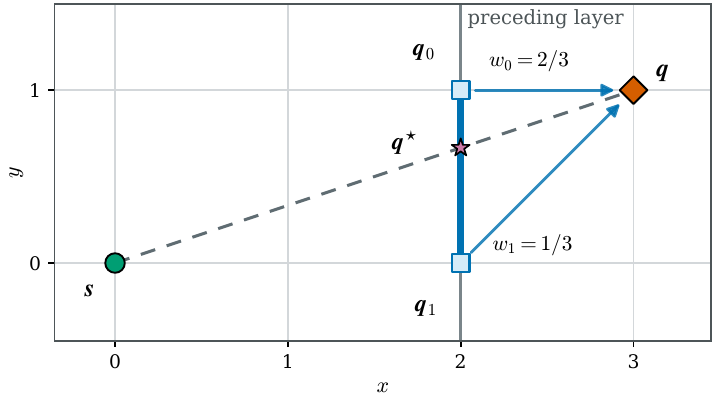}
	\caption{Local two-dimensional free-transport stencil.  The source--target
		ray meets the preceding layer at
		$\boldsymbol{q}^{\star}=(2,2/3)$, between predecessors
		$\boldsymbol{q}_0=(2,1)$ and $\boldsymbol{q}_1=(2,0)$.  Their weights
		$2/3$ and $1/3$ give the interpolation in
		Eq.~\eqref{eq:2d-interpolant}.  The target transmission
		$\tau(\boldsymbol{q})$ is applied afterward in
		Eq.~\eqref{eq:visibility-recurrence}.}
	\label{fig:stencil-geometry-2d}
\end{figure*}

Our earlier work \cite{IbrahimEtAl2024Visibility} used the two-dimensional
free-transport interpolant in Eq.~\eqref{eq:2d-interpolant}.  This interpolant
corresponds to the source-relative operator in
Eq.~\eqref{eq:radial-transport}, although the
paper's displayed transport equation and accompanying scalar/upwind derivation
interchanged the coordinate components and derivatives.

In 3D, suppose $x$ is dominant and put $p_y=|n_y|/|n_x|$ and
$p_z=|n_z|/|n_x|$.  Define
\[
	\boldsymbol{q}_{ab}=\boldsymbol{q}-h\sigma_x\boldsymbol{e}_x
	-ah\sigma_y\boldsymbol{e}_y-bh\sigma_z\boldsymbol{e}_z,
	\qquad a,b\in\{0,1\}.
\]
The four-parent free-transport interpolant is
\begin{align}
	\mathcal{I}_hV(\boldsymbol{q})={} &
	(1-p_y)(1-p_z)V(\boldsymbol{q}_{00})
	+p_y(1-p_z)V(\boldsymbol{q}_{10}) \notag                              \\
	                                  & +(1-p_y)p_zV(\boldsymbol{q}_{01})
	+p_yp_zV(\boldsymbol{q}_{11}).
	\label{eq:3d-interpolant}
\end{align}
The other dominant coordinates follow by permutation.

Degenerate transverse-coordinate ratios require no separate numerical formula.
If $n_j=0$, then $p_j=0$ and only the corner that does not move in coordinate
$j$ has nonzero weight.  If $|n_j|=M$, coordinate $j$ ties the dominant
magnitude, so $p_j=1$ and every nonzero predecessor moves inward in that
coordinate as well.  An axis-aligned target therefore has one axial
predecessor, a 2D diagonal target has one diagonal predecessor, and a 3D
space-diagonal target has one space-diagonal predecessor.  On a 3D face
diagonal, both tied coordinates move together, and only a strictly smaller
third coordinate is interpolated.  Choosing a different tied dominant
coordinate produces the same set of nonzero-weight predecessors, so the
$x$--$y$--$z$ priority selects one tied coordinate as dominant for each target
without changing the interpolated value.

Attenuation is applied only after the free-transport interpolation.  The full
recurrence is
\begin{equation}
	V(\boldsymbol{s})=V_s,
	\qquad
	V(\boldsymbol{q})=\tau(\boldsymbol{q})
	\mathcal{I}_hV(\boldsymbol{q})
	\quad(\boldsymbol{q}\ne\boldsymbol{s}).
	\label{eq:visibility-recurrence}
\end{equation}
The factor $\tau$ models discrete attenuation separately from the
free-transport PDE~\eqref{eq:radial-transport} and does not define a continuous
optical equation.

\subsection{Causality, boundedness, and local PDE consistency}

The visibility recurrence above is useful only if the weights in its
free-transport interpolant form a valid average, every predecessor with
nonzero weight has already been computed, and the interpolant remains
consistent with the radial PDE.  The following proposition establishes the
averaging and causality properties, after which a Taylor expansion verifies
local PDE consistency.

\begin{proposition}[Recurrence invariants]
	For $d\in\{2,3\}$, the weights in~\eqref{eq:predecessor-weights} are
	nonnegative and sum to one.  Every predecessor with nonzero weight lies on
	layer $M-1$.  Consequently, the recurrence is causal.  If $V_s\in[0,1]$ and
	$\tau\in[0,1]$, then $V\in[0,1]$.  When $\tau\equiv1$, the initial visibility
	value $V_s$ is transported throughout the grid.
\end{proposition}

\begin{proof}
	Each $p_j$ belongs to $[0,1]$, so every factor $\omega_j$ is nonnegative, and
	\[
		\sum_{\boldsymbol{\epsilon}}w_{\boldsymbol{\epsilon}}
		=\prod_{j\ne m}\bigl((1-p_j)+p_j\bigr)=1.
	\]
	The $m$th source-relative coordinate of every labelled point
	$\boldsymbol{q}_{\boldsymbol{\epsilon}}$ has magnitude $M-1$.  Any other
	source-relative coordinate tied at magnitude $M$ has
	$p_j=1$, so a nonzero-weight corner must move that coordinate inward too.  All
	remaining coordinates already have magnitude at most $M-1$.  Thus every
	predecessor lies on the preceding layer.
	Causality follows, and induction over $M$ proves boundedness.  The same
	induction preserves a constant value when $\tau\equiv1$.
	\qed
\end{proof}

The free-transport interpolant~\eqref{eq:multilinear-interpolant} is also
locally consistent with the radial equation~\eqref{eq:radial-transport} from
which the interpolant was derived.  For a smooth $u$, Taylor expansion about
$\boldsymbol{q}$ gives
\begin{equation}
	\mathcal{I}_hu(\boldsymbol{q})
	=u(\boldsymbol{q})-h\left(
	\sigma_m u_{x_m}
	+\sum_{j\ne m}\sigma_jp_j u_{x_j}\right)+O(h^2).
	\label{eq:taylor-consistency}
\end{equation}
All derivatives in~\eqref{eq:taylor-consistency} are evaluated at
$\boldsymbol{q}$.  The expression in parentheses equals
$(\boldsymbol{q}-\boldsymbol{s})\mathbin{\cdot}\nabla
	u/|q_m-s_m|$.  Hence $(u-\mathcal{I}_hu)/h$ is a first-order approximation of
the scaled radial operator inside each dominant-coordinate region.  For a
smooth ray-constant field, the radial operator vanishes, leaving second-order
local interpolation error at grid points distinct from the source.  The
first-order operator consistency and second-order ray-constant interpolation
error concern the free-transport stencil and therefore do not, by themselves,
establish exact agreement with continuous visibility on a fixed grid.  Theorem~\ref{thm:nongrazing-refinement-consistency} later establishes when
thresholded recurrence classifications agree with continuous visibility under
grid refinement.

The source-ray derivation identifies the predecessor locations and weights,
while the invariant proof establishes that the outward sweep is causal and
preserves the value bounds.  Neither result explains what a visibility value
strictly between zero and one measures.  Explaining a fractional visibility
value requires identifying the source-to-target paths whose transmission
products the recurrence averages and showing why the proportion of paths
entering through each of the target's two predecessors equals the corresponding
linear weight in
Eq.~\eqref{eq:2d-interpolant}.

\section{Exact Two-Dimensional Path Model}
\label{sec:exact-2d-path-model}

The visibility recurrence derived in Section~\ref{sec:recurrence} computes the
scalar visibility value $V(\boldsymbol{q})$ at each target cell from predecessors
on the immediately preceding layer.  Repeatedly
substituting those predecessor values must terminate at the source cell, where
the initial visibility value is $V_s$, because every substitution reduces the
layer index by one.  Reading the same
dependencies from the source toward the target turns each possible sequence of
predecessor choices into a grid path.  This section defines the resulting finite
family of grid paths and proves that, for every target cell $\boldsymbol{q}$, the scalar
visibility value $V(\boldsymbol{q})$ equals $V_s$ times the mean of the products
of the cell transmissions encountered along the
paths in that family.

\subsection{Canonical paths from layer choices}

Consider the two-dimensional case and retain the source-relative notation of
Section~\ref{sec:recurrence}.  For a target
$\boldsymbol{q}=\boldsymbol{s}+h\boldsymbol{n}$ distinct from the source, let
$m$ be its dominant-coordinate index and let $r$ denote the other, transverse
coordinate.  The dominant displacement is $M=|n_m|$, while
\begin{equation}
	R=|n_r|,
	\qquad 0\leq R\leq M,
	\label{eq:transverse-advance-count}
\end{equation}
is the number of transverse advances required to reach the target.  Every path
makes one signed dominant-coordinate step
$h\sigma_m\boldsymbol{e}_m$ on each of its $M$ ordered steps.  Exactly $R$ of
those steps also include the signed transverse displacement
$h\sigma_r\boldsymbol{e}_r$.

Choose the $R$ steps that advance transversely by a subset
$S\subseteq\{1,\ldots,M\}$ with $|S|=R$.  After $\ell$ steps, the resulting
path is at
\begin{equation}
	\begin{aligned}
		\boldsymbol{q}_{\ell}(S)
		&=\boldsymbol{s}+h\left(
		\ell\sigma_m\boldsymbol{e}_m
		+|S\cap\{1,\ldots,\ell\}|\sigma_r\boldsymbol{e}_r
		\right),\\
		&\hspace{7.5em}0\leq\ell\leq M.
	\end{aligned}
	\label{eq:canonical-2d-path-point}
\end{equation}
The subset $S$ is the path's \emph{transverse-advance schedule}.  It uniquely
determines the sequence
$\pi_S=(\boldsymbol{q}_1(S),\ldots,\boldsymbol{q}_M(S))$, which records the
cells reached after the recurrence is initialized at the source and therefore
excludes the source while including the target.  The canonical path family is
therefore
\begin{equation}
	\begin{aligned}
		\mathcal{P}(\boldsymbol{s},\boldsymbol{q})
		&=\left\{\pi_S:
		S\subseteq\{1,\ldots,M\},\ |S|=R\right\},\\
		|\mathcal{P}(\boldsymbol{s},\boldsymbol{q})|
		&=\binom{M}{R}.
	\end{aligned}
	\label{eq:canonical-2d-path-family}
\end{equation}
For the coincident endpoint $\boldsymbol{q}=\boldsymbol{s}$, the family
contains one empty path.  An axis target has $R=0$ and therefore one
axis-aligned canonical path, whereas a diagonal target has $R=M$ and one path
on which both coordinates advance at every step.  Translation, coordinate reflection,
and exchange of the dominant coordinate give bijections between the remaining
signed cases.  Because each spatial coordinate along a path varies monotonically
between the corresponding coordinates of the source and target, every path
remains inside any rectangular grid that contains both endpoints.

Figure~\ref{fig:canonical-paths-2d} shows all six canonical paths in a small
example.

\begin{figure*}[t]
	\centering
	\includegraphics[width=\textwidth]{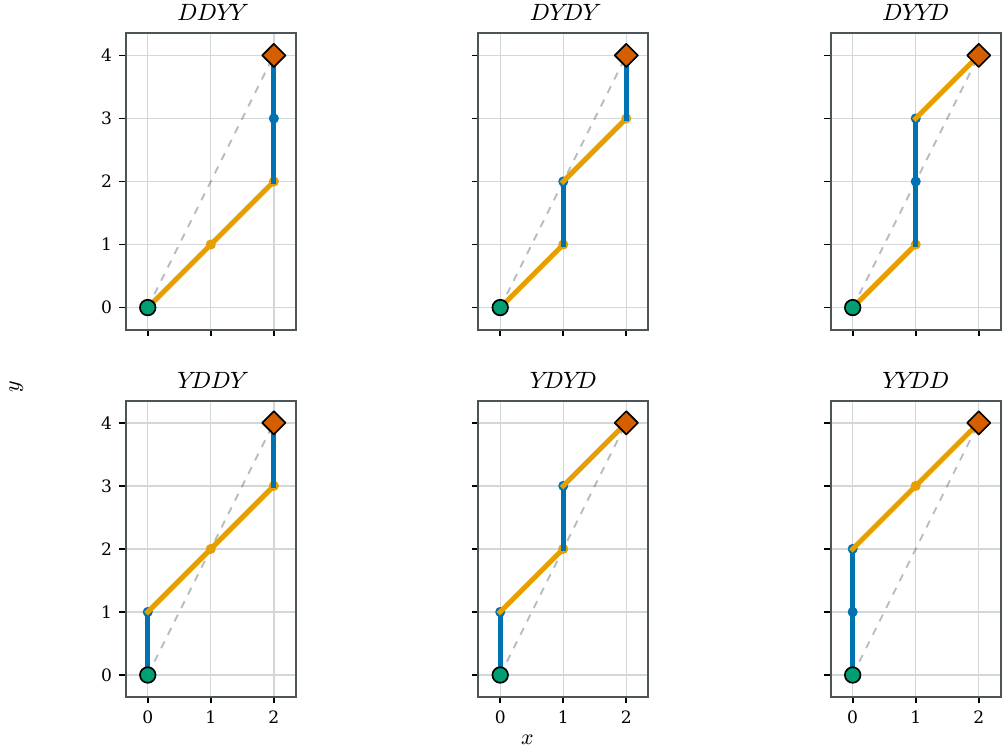}
	\caption{The six canonical paths from $\boldsymbol{s}=(0,0)$ to
		$\boldsymbol{q}=(2,4)$.  Blue segments are dominant-only steps $Y$,
		orange segments are diagonal steps $D$, and each four-letter word records
		the order of the $Y$ and $D$ steps.
		The dashed line shows the direct source--target segment and is not an
		additional member of the canonical path family.}
	\label{fig:canonical-paths-2d}
\end{figure*}

For this example, $h=1$, $\boldsymbol{s}=(0,0)$, and
$\boldsymbol{q}=(2,4)$.  The $y$ coordinate is dominant, so $M=4$ and $R=2$.
Write $Y=(0,1)$ for a dominant-only step and $D=(1,1)$ for a step that also
advances transversely.  Choosing two of the four step positions for $D$ gives
\begin{equation}
	\begin{aligned}
		\mathcal{P}((0,0),(2,4))
		=\{&DDYY,DYDY,DYYD,\\
		&YDDY,YDYD,YYDD\}.
	\end{aligned}
	\label{eq:six-path-words}
\end{equation}

\subsection{Uniform finite-path average}

For the canonical path $\pi_S$ determined by a transverse-advance schedule
$S$, define its transmission product by
\begin{equation}
	\mathcal{T}(\pi_S)
	=\prod_{\ell=1}^{M}\tau\bigl(\boldsymbol{q}_{\ell}(S)\bigr).
	\label{eq:path-transmission-product}
\end{equation}
Because $\pi_S$ lists the cells reached after source initialization,
$\mathcal{T}(\pi_S)$ contains every non-source attenuation factor along that
path, including the target's, while the initial visibility value $V_s$
multiplies the average below.  For a coincident endpoint, the path is empty and its
transmission product is one, so the aggregate equals $V_s$.  For
$\boldsymbol{q}\ne\boldsymbol{s}$, averaging the transmission products
uniformly across the canonical path family gives
\begin{equation}
	V_{\mathcal{P}}(\boldsymbol{q})
	=\frac{V_s}{\binom{M}{R}}
	\sum_{\substack{S\subseteq\{1,\ldots,M\}\\|S|=R}}
	\mathcal{T}(\pi_S),
	\qquad
	V_{\mathcal{P}}(\boldsymbol{s})=V_s.
	\label{eq:canonical-path-average}
\end{equation}
The schedules and canonical paths correspond one-to-one, so selecting the
schedules uniformly gives every path the same probability
$1/\binom{M}{R}$, although different paths may contribute different
transmission products.

For the six-path example, put
$\tau_{ij}=\tau((i,j))$.  Factoring out the target transmission makes the
complete expansion compact:
\begin{align}
	V_{\mathcal{P}}((2,4))
	={} & \frac{V_s\tau_{24}}{6}\bigl[
		\tau_{11}\tau_{22}\tau_{23}
		+\tau_{11}\tau_{12}\tau_{23}
	+\tau_{11}\tau_{12}\tau_{13} \notag \\
	    & \qquad
		+\tau_{01}\tau_{12}\tau_{23}
		+\tau_{01}\tau_{12}\tau_{13}
		+\tau_{01}\tau_{02}\tau_{13}
		\bigr].
	\label{eq:six-path-transmission-expansion}
\end{align}
The six products occur in the same order as the path words in
Eq.~\eqref{eq:six-path-words}.  The causal sweep does not construct this
expansion explicitly.  It appears when the dynamic program is unrolled for
one target: the recurrence is expanded by repeatedly substituting each
predecessor value until every branch reaches the source-cell base case $V_s$.

\begin{theorem}[Path-average identity]
	\label{thm:exact-2d-path-representation}
	For any source and target in a rectangular two-dimensional grid, any
	$V_s\in[0,1]$, and any deterministic transmission field
	$\tau:\mathcal{G}_h\to[0,1]$, the dominant-axis recurrence in
	Eq.~\eqref{eq:visibility-recurrence} satisfies
	$V(\boldsymbol{q})=V_{\mathcal{P}}(\boldsymbol{q})$.
\end{theorem}

\begin{proof}
	The coincident case is the prescribed boundary value.  For $M>0$, divide the
	schedules in Eq.~\eqref{eq:canonical-2d-path-family} according to whether
	the transverse coordinate advances on step $M$, the final path step into the
	target---that is, whether $M\in S$.  The $\binom{M-1}{R}$ schedules that do
	not contain $M$ reach the target through the predecessor with no transverse
	advance
	\begin{equation}
		\boldsymbol{q}_A
		=\boldsymbol{q}-h\sigma_m\boldsymbol{e}_m,
		\label{eq:axis-only-final-predecessor}
	\end{equation}
	whereas the $\binom{M-1}{R-1}$ schedules that contain $M$ reach it through
	the diagonal predecessor
	\begin{equation}
		\boldsymbol{q}_D
		=\boldsymbol{q}-h\sigma_m\boldsymbol{e}_m
		-h\sigma_r\boldsymbol{e}_r.
		\label{eq:diagonal-final-predecessor}
	\end{equation}
	When $R<M$, removing the target from every path in the first group gives a
	bijection onto $\mathcal{P}(\boldsymbol{s},\boldsymbol{q}_A)$.  When $R>0$,
	the same removal maps the second group bijectively onto
	$\mathcal{P}(\boldsymbol{s},\boldsymbol{q}_D)$.  The first group is empty for
	$R=M$, and the second is empty for $R=0$.  We therefore use the standard
	convention that a binomial coefficient is zero when its lower index lies
	outside $\{0,\ldots,M-1\}$.  If a predecessor has equal coordinate
	displacements, exchanging its dominant coordinate leaves its canonical path
	family unchanged.

	Every path transmission product also contains the common factor
	$\tau(\boldsymbol{q})$.  Applying the relevant bijections and using
	\begin{equation}
		\frac{\binom{M-1}{R}}{\binom{M}{R}}=\frac{M-R}{M},
		\qquad
		\frac{\binom{M-1}{R-1}}{\binom{M}{R}}=\frac{R}{M},
		\label{eq:final-step-path-proportions}
	\end{equation}
	therefore gives
	\begin{equation}
		V_{\mathcal{P}}(\boldsymbol{q})
		=\tau(\boldsymbol{q})\left[
		\left(1-\frac{R}{M}\right)V_{\mathcal{P}}(\boldsymbol{q}_A)
		+\frac{R}{M}V_{\mathcal{P}}(\boldsymbol{q}_D)
		\right].
		\label{eq:path-average-recurrence}
	\end{equation}
	For an $x$-dominant target, $\boldsymbol{q}_A$ and
	$\boldsymbol{q}_D$ are the predecessors $\boldsymbol{q}_0$ and
	$\boldsymbol{q}_1$ in Eq.~\eqref{eq:2d-interpolant}.  The $y$-dominant case
	exchanges the coordinate roles.  Because $R/M$ is the transverse
	interpolation ratio, the bracketed expression in
	Eq.~\eqref{eq:path-average-recurrence} is exactly the two-parent
	free-transport interpolant in Eq.~\eqref{eq:2d-interpolant}.  The factor
	$\tau(\boldsymbol{q})$ outside the brackets then applies target-cell
	attenuation, so Eq.~\eqref{eq:path-average-recurrence} matches the complete
	two-dimensional visibility recurrence in
	Eq.~\eqref{eq:visibility-recurrence}.  The empty subfamily in an axis or
	diagonal case has zero coefficient, so every predecessor with nonzero
	coefficient lies on layer $M-1$ and contributes the finite-path average at
	that predecessor.
	Induction over the layer index $M$ proves the result.
	\qed
\end{proof}

\subsection{Uniform path selection and map uncertainty}

Theorem~\ref{thm:exact-2d-path-representation} establishes that the scalar
visibility value $V(\boldsymbol{q})$ equals $V_{\mathcal{P}}(\boldsymbol{q})$,
where Eq.~\eqref{eq:canonical-path-average} defines
$V_{\mathcal{P}}(\boldsymbol{q})$ as the uniform average of the transmission
products over the canonical path family.  If $\Pi$ denotes a
canonical path selected uniformly from
$\mathcal{P}(\boldsymbol{s},\boldsymbol{q})$, the resulting path-average
identity can be written as
\begin{equation}
	V(\boldsymbol{q})
	=V_s\,\mathbb{E}\!\left[\mathcal{T}(\Pi)\right].
	\label{eq:structural-path-expectation}
\end{equation}
For deterministic transmissions,
Eq.~\eqref{eq:structural-path-expectation} states that
$V(\boldsymbol{q})$ is $V_s$ times the mean of
the path transmission products over the canonical path family.
The randomness comes only from selecting $\Pi$; the cell transmissions remain
prescribed inputs.  Under the unit-source convention, binary cell transmissions
make $\mathcal{T}(\Pi)=1$ precisely when
every cell visited after the source is free, so $V(\boldsymbol{q})$ is the
fraction of the canonical path family that remains clear.  It is not the
indicator or probability that at least one canonical path is clear.
Exhaustive enumeration evaluates every path in the canonical family directly
rather than executing the causal sweep.  Agreement between the enumerated path
average and the scalar visibility value returned by the exact-arithmetic 2D
evaluator provides an independent check that the implementation satisfies
Eq.~\eqref{eq:structural-path-expectation}.

Goldstein et al.~\cite{Goldstein2020SpaceAnalysis,Goldstein2022PathCounting}
proved a closely related binary 2D result: the visibility recurrence they
analyse returns the fraction of shortest source-to-target grid paths that
remain unobstructed.  The formulation developed here differs in two ways.
First, Section~\ref{sec:recurrence} writes the recurrence relative to an
arbitrary source position and direction, so source translation, coordinate
signs, and tie cases are handled within one formulation.  Second,
Eq.~\eqref{eq:visibility-recurrence} replaces binary predecessor--target edge
tests with a prescribed transmission value in $[0,1]$ at the destination cell.
Each canonical path therefore contributes the product of the cell
transmissions it encounters, as defined in
Eq.~\eqref{eq:path-transmission-product}.
Theorem~\ref{thm:exact-2d-path-representation} proves that the recurrence
returns the uniform average of these path-transmission products.  For binary
transmissions, the fraction of canonical paths that remain clear is
consequently one case of a more general deterministic path-transmission
interpretation.
Section~\ref{sec:hypergeometric-bridge} uses this 2D identity to derive the
probability that a uniformly selected path passes through each intermediate
cell.  Section~\ref{sec:multidimensional-extension} then proves the
corresponding recurrence--path-average identity in 3D and, afterward, in
arbitrary fixed spatial dimension $d$.

If the input field instead assigns each cell its individual probability of
being free under a random-map model, multiplying those probabilities along a
path represents the path's expected transmission only when the required joint
moments factorize, as stated in Section~\ref{sec:recurrence}.

The uniform finite-path average in Eq.~\eqref{eq:structural-path-expectation}
explains how complete source-to-target paths
combine to produce the scalar visibility value $V(\boldsymbol{q})$, but it does
not yet say how those paths are distributed among the cells on each
intermediate layer.  At a fixed intermediate layer,
different canonical paths can pass through different cells according to how
many of the $R$ required transverse advances have already occurred.  The
remaining counting question is: what fraction of the canonical path family passes
through each possible cell on that layer?

\section{Endpoint-Conditioned Hypergeometric Bridge}
\label{sec:hypergeometric-bridge}

The closing question of Section~\ref{sec:exact-2d-path-model} concerns the same
canonical path family at one layer rather than along a complete source--target
path.  Retain a target distinct from the source and the notation $M$, $R$,
$S$, and $\ell$ established there.  Select the transverse-advance schedule
$S$ uniformly from the $\binom{M}{R}$ admissible schedules, which is equivalent
to selecting the corresponding canonical path uniformly.  For
$0\leq\ell\leq M$, define the prefix count
\begin{equation}
	K_{\ell}=\left|S\cap\{1,\ldots,\ell\}\right|.
	\label{eq:bridge-prefix-count}
\end{equation}
The count $K_{\ell}$ is the number of transverse advances completed during the
first $\ell$ ordered path steps from the source toward the target.  If
$K_{\ell}=k$, the cell reached after those steps is
\begin{equation}
	\boldsymbol{q}_{\ell,k}
	=\boldsymbol{s}+h\left(
	\ell\sigma_m\boldsymbol{e}_m
	+k\sigma_r\boldsymbol{e}_r
	\right),
	\qquad
	\boldsymbol{q}_{\ell}(S)=\boldsymbol{q}_{\ell,K_{\ell}}.
	\label{eq:bridge-layer-cell}
\end{equation}
For the bridge geometry, $\boldsymbol{q}_{0,0}=\boldsymbol{s}$ records the
fixed layer-zero endpoint.  The canonical path and its transmission product
remain defined as in Section~\ref{sec:exact-2d-path-model}: they begin after
source initialization and therefore do not acquire a source-cell transmission
factor.

\subsection{Incidence on a fixed layer}

At a fixed positive layer $\ell$, a canonical path passes through
$\boldsymbol{q}_{\ell,k}$ exactly when its schedule has selected $k$ of the
first $\ell$ step indices.  Such a schedule chooses those $k$ indices from the
prefix and the remaining $R-k$ transverse advances from the $M-\ell$ indices
after that layer.  The fraction of admissible schedules for which the
corresponding path reaches $\boldsymbol{q}_{\ell,k}$ is therefore
\begin{equation}
	H_{\ell,k}
	:=\Pr\!\left(K_{\ell}=k\right)
	=\frac{\binom{\ell}{k}\binom{M-\ell}{R-k}}
	{\binom{M}{R}},
	\label{eq:hypergeometric-bridge-mass}
\end{equation}
on the feasible support
\begin{equation}
	\max\{0,R-(M-\ell)\}
	\leq k\leq
	\min\{\ell,R\}.
	\label{eq:hypergeometric-bridge-support}
\end{equation}
Outside this range, $H_{\ell,k}=0$.  The events indexed by the support exhaust
the admissible schedules, so their probabilities sum to one.  In ordinary
sampling language, the $M$ step indices form a finite population containing
exactly $R$ transverse-advance indices, and $K_{\ell}$ counts how many occur in
the first $\ell$ draws without replacement.  Conditioning the total count at
the target to equal $R$ makes Eq.~\eqref{eq:hypergeometric-bridge-mass} an
endpoint-conditioned discrete hypergeometric distribution, which we call the
\emph{hypergeometric bridge}.  Its value $H_{\ell,k}$ is the
\emph{path-incidence probability}: for $1\leq\ell\leq M$, it is the probability
that a uniformly selected canonical path passes through
$\boldsymbol{q}_{\ell,k}$.  At $\ell=0$, $H_{0,0}=1$ instead records the fixed
source endpoint of the bridge.

Figure~\ref{fig:path-incidence-2d} evaluates the bridge on the six-path example
from Section~\ref{sec:exact-2d-path-model}.

\begin{figure*}[t]
	\centering
	\includegraphics[width=\textwidth]{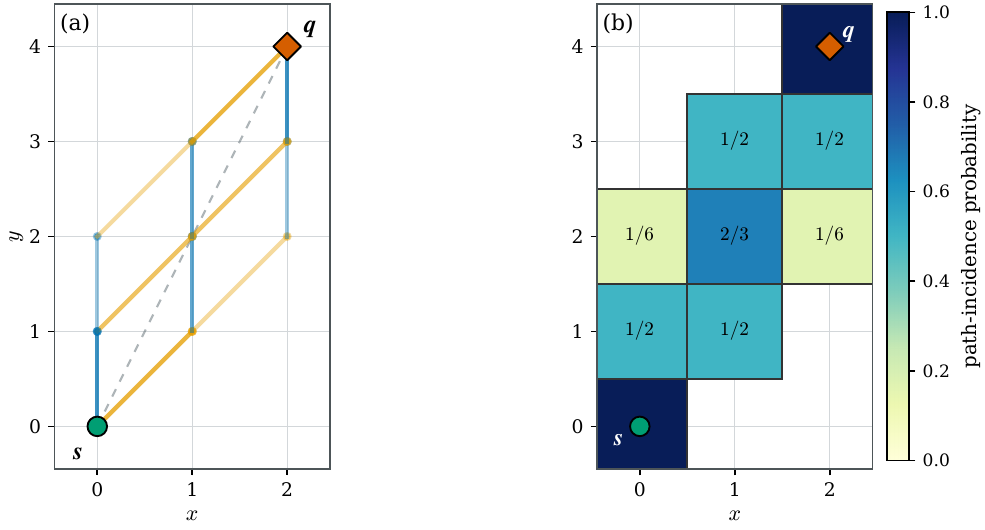}
	\caption{Path-incidence probabilities for the six canonical paths from
		$\boldsymbol{s}=(0,0)$ to $\boldsymbol{q}=(2,4)$.  Panel (a) overlays the
		paths shown separately in Fig.~\ref{fig:canonical-paths-2d}; blue segments
		are dominant-only steps and orange segments also advance transversely.
		Panel (b) assigns each reached cell the fraction of the six paths passing
		through it.  The source and target layers are point masses with probability
		one.  Showing the source probability records the bridge endpoint and does
		not include source-cell transmission in the path product.}
	\label{fig:path-incidence-2d}
\end{figure*}

Goldstein et al.~\cite{Goldstein2022PathCounting} count the same endpoint-fixed
paths through an intermediate cell by multiplying the numbers of paths before
and after that cell.  Dividing the resulting product by the total number of
source--target paths gives exactly Eq.~\eqref{eq:hypergeometric-bridge-mass}.
Identifying the resulting path-incidence probability as a hypergeometric bridge
makes its exact support, moments, and endpoint pinching explicit and connects
those properties to the scalar and multidimensional visibility model.

\subsection{Mean location, transverse spread, and fixed endpoints}

The bridge moments describe how the canonical path family is distributed
across an intermediate layer.  Write $p=R/M$ for the fraction of the $M$ step
indices at which a path advances transversely.  Because the schedule is
selected uniformly from all subsets of size $R$, any specified step index is
selected with probability $p$.  For each $i\in\{1,\ldots,M\}$, let $I_i=1$ if
$i\in S$ and $I_i=0$ otherwise.  Thus $I_i$ records whether step $i$ includes
a transverse advance, and $K_{\ell}=\sum_{i=1}^{\ell}I_i$ counts the transverse
advances completed during the first $\ell$ steps.  Since
$\mathbb{E}[I_i]=p$, linearity of expectation gives
\begin{equation}
	\mathbb{E}[K_{\ell}]=\frac{\ell R}{M}=\ell p.
	\label{eq:hypergeometric-bridge-mean}
\end{equation}
For $M>1$ and $i\ne j$, the product $I_iI_j$ equals one exactly when both
indices belong to $S$.  Fixed-cardinality sampling therefore gives
$\mathbb{E}[I_iI_j]=\Pr(i,j\in S)=R(R-1)/(M(M-1))$.  The resulting covariance,
$\operatorname{Cov}(I_i,I_j)=-p(1-p)/(M-1)$, is negative because selecting one
index leaves only $R-1$ transverse advances among the other $M-1$ indices.
Summing the $\ell$ indicator variances and their pairwise covariances gives
\begin{equation}
	\operatorname{Var}(K_{\ell})
	=\ell p(1-p)\frac{M-\ell}{M-1}.
	\label{eq:hypergeometric-bridge-variance}
\end{equation}
When $M=1$, the distribution is a point mass and its variance is zero.

Mapping the mean count through Eq.~\eqref{eq:bridge-layer-cell} gives
\begin{equation}
	\mathbb{E}\!\left[\boldsymbol{q}_{\ell,K_{\ell}}\right]
	=\boldsymbol{s}+\frac{\ell}{M}
	\left(\boldsymbol{q}-\boldsymbol{s}\right).
	\label{eq:bridge-mean-cell}
\end{equation}
Thus the mean of the cells reached on layer $\ell$ lies on the source--target
segment, while Eq.~\eqref{eq:hypergeometric-bridge-variance} quantifies the
finite transverse spread about that mean.  At $\ell=0$, every schedule has
$K_0=0$, and at $\ell=M$, every schedule has $K_M=R$; the bridge therefore
collapses to the source and target at its two endpoints.  This collapse is the
bridge's \emph{endpoint pinching}.  An axis target, for which $R=0$, and a
diagonal target, for which $R=M$, are point masses on every layer.  The
coincident source and target similarly give the single layer-zero point mass.
The mean alignment and endpoint pinching describe the declared canonical path
family; neither property makes the thresholded field agree exactly with
continuous visibility at a fixed grid resolution.

Figure~\ref{fig:bridge-field-2d} shows the same bridge geometry for a larger
source--target displacement.

\begin{figure*}[t]
	\centering
	\includegraphics[width=\textwidth]{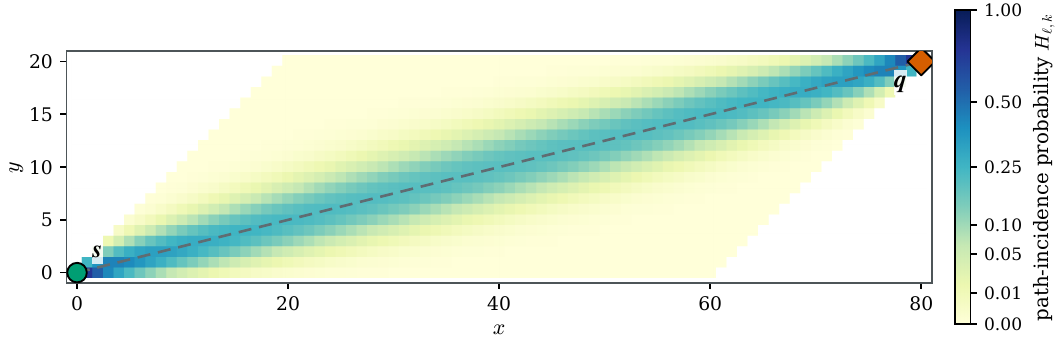}
	\caption{Endpoint-conditioned path-incidence field from
		$\boldsymbol{s}=(0,0)$ to $\boldsymbol{q}=(80,20)$, with $M=80$ and
		$R=20$.  Each coloured grid cell represents its exact probability
		$H_{\ell,k}$.  The dashed source--target segment passes through the mean
		layer location in Eq.~\eqref{eq:bridge-mean-cell}.  Colour intensity is
		scaled by $\sqrt{H_{\ell,k}}$ to keep low-probability cells visible, while
		the colour bar reports $H_{\ell,k}$ itself.  The widening interior support
		and unit masses at the two endpoints show the bridge's transverse spread
		and endpoint pinching.}
	\label{fig:bridge-field-2d}
\end{figure*}

\subsection{Consequences for transmission and thresholding}

For the fixed source-to-target family
$\mathcal{P}(\boldsymbol{s},\boldsymbol{q})$, an intermediate cell's
path-incidence probability directly weights the target visibility
$V(\boldsymbol{q})$ in the single-cell and single-layer transmission settings
derived below.  For a non-source cell $\boldsymbol{p}$, let
\begin{equation}
	A_{\boldsymbol{p}}
	=\{\pi\in\mathcal{P}(\boldsymbol{s},\boldsymbol{q}):
	\boldsymbol{p}\in\pi\},
	\qquad
	H_{\boldsymbol{p}}=\Pr(\Pi\in A_{\boldsymbol{p}}).
	\label{eq:path-incidence-event}
\end{equation}
Thus $A_{\boldsymbol{p}}$ is the event that the uniformly selected canonical
path $\Pi$ visits $\boldsymbol{p}$, and $H_{\boldsymbol{p}}$ is the fraction of
source-to-$\boldsymbol{q}$ canonical paths that do so.  The incidence
probability $H_{\boldsymbol{p}}$ is distinct from $V(\boldsymbol{p})$, which is
$V_s$ times the mean of the transmission products over the
source-to-$\boldsymbol{p}$ canonical path family
$\mathcal{P}(\boldsymbol{s},\boldsymbol{p})$.
For a cell $\boldsymbol{p}=\boldsymbol{q}_{\ell,k}$, this probability is
$H_{\ell,k}$.  Suppose every non-source cell is fully transmitting except
$\boldsymbol{p}$, whose transmission is $a\in[0,1]$.  Separating paths that
visit $\boldsymbol{p}$ from those that do not in
Eq.~\eqref{eq:structural-path-expectation} gives the exact identity
\begin{equation}
	V(\boldsymbol{q})
	=V_s\left[1-(1-a)H_{\boldsymbol{p}}\right].
	\label{eq:single-cell-incidence-response}
\end{equation}
More generally, if only one intermediate layer $\mathcal{L}$ contains nonunit
transmissions, every path visits exactly one cell on that layer and
\begin{equation}
	V(\boldsymbol{q})
	=V_s\sum_{\boldsymbol{p}\in\mathcal{L}}
	H_{\boldsymbol{p}}\tau(\boldsymbol{p}).
	\label{eq:heterogeneous-layer-response}
\end{equation}
Equations~\eqref{eq:single-cell-incidence-response} and
\eqref{eq:heterogeneous-layer-response} give the path-incidence probabilities a
direct transmission meaning.  When every other non-source cell is fully transmitting,
reducing $\tau(\boldsymbol{p})$ from one to $a$ reduces
$V(\boldsymbol{q})/V_s$ by exactly $(1-a)H_{\boldsymbol{p}}$ whenever $V_s>0$.
When nonunit transmissions are confined to one layer, the incidence
probabilities $H_{\boldsymbol{p}}$ are the exact weights by which that layer's
cell transmissions combine to produce $V(\boldsymbol{q})/V_s$.
In a general transmission field, the effect of changing one cell also depends
on the transmissions encountered with it along the same paths, so
path-incidence probability alone is not a universal sensitivity coefficient.

For a fixed binary map under the unit-source convention, let $B$ be its set of
blocked non-source cells.  A path contributes one exactly when it avoids every
cell in $B$, and hence
\begin{equation}
	V(\boldsymbol{q})
	=1-\Pr\!\left(\bigcup_{\boldsymbol{p}\in B}
	A_{\boldsymbol{p}}\right).
	\label{eq:blocker-union-identity}
\end{equation}
The union is the event that the uniformly selected canonical path $\Pi$ visits
at least one blocked cell, so Eq.~\eqref{eq:blocker-union-identity} is the
probability that $\Pi$ avoids every blocked cell.
For a single blocked cell $\boldsymbol{p}$,
Eq.~\eqref{eq:blocker-union-identity} reduces to
$V(\boldsymbol{q})=1-H_{\boldsymbol{p}}$, the fraction of canonical paths that
avoid $\boldsymbol{p}$.  If all blockers lie on one layer, their path events
are disjoint because every path reaches exactly one cell on that layer, so the
blocked probability in Eq.~\eqref{eq:blocker-union-identity} is the sum of
their incidence probabilities.  Events on different layers can overlap along
the same path and must not in general be added.

Under the threshold convention established in Section~\ref{sec:recurrence}, a
target is visible when $V(\boldsymbol{q})\geq\theta$.  If the fixed binary map
leaves $C$ of the
$|\mathcal{P}(\boldsymbol{s},\boldsymbol{q})|$ canonical paths clear, this rule
is equivalent to
\begin{equation}
	C\geq\left\lceil
	\theta|\mathcal{P}(\boldsymbol{s},\boldsymbol{q})|
	\right\rceil.
	\label{eq:threshold-clear-path-count}
\end{equation}
At $\theta=1/2$, at least half of the canonical paths must remain clear, and an
exact tie is visible.  For blockers confined to one layer,
Eq.~\eqref{eq:blocker-union-identity} makes the target invisible exactly when
the sum of their incidence probabilities is strictly greater than $1-\theta$.
With nonbinary cell transmissions, however, the same threshold acts on the
mean of the path transmission products rather than on a count of clear paths.
In particular, at finite resolution, classifying the target as visible when
$V(\boldsymbol{q})\geq 1/2$ need not match the continuous-visibility
classification of the straight source--target segment.

The path-incidence probabilities arise from uniform selection within the fixed
canonical path family, not from uncertainty in the map.  If a map instead
supplies jointly distributed transmission variables $T(\boldsymbol{p})$, write $V(T)$ for the
scalar visibility value obtained from one realization of that field.  Under
the unit-source convention, writing
$\mathcal{P}=\mathcal{P}(\boldsymbol{s},\boldsymbol{q})$ gives
\begin{equation}
	\mathbb{E}_{\mathrm{map}}[V(T)]
	=\frac{1}{|\mathcal{P}|}
	\sum_{\pi\in\mathcal{P}}
	\mathbb{E}_{\mathrm{map}}\!\left[
		\prod_{\boldsymbol{p}\in\pi}T(\boldsymbol{p})\right].
	\label{eq:random-map-path-moments}
\end{equation}
Equation~\eqref{eq:random-map-path-moments} states that the expected target
visibility over map realizations is the average, over the canonical path
family, of the expected product of the cell transmissions along each path.
The transmitting-source convention includes $T(\boldsymbol{s})$ once in every
product.  Replacing each $T(\boldsymbol{p})$ by its individual mean
$\mathbb{E}_{\mathrm{map}}[T(\boldsymbol{p})]$ is exact for each path only when
the expected product of the transmissions along that path equals the product
of their individual means.  This equality is the required factorization.
Mutual independence along each path is sufficient but not necessary.  The role
of map correlation is clear even for two cells $\boldsymbol{p}_1$ and
$\boldsymbol{p}_2$ on an otherwise fully transmitting axis path.  Let their
transmissions be perfectly correlated,
$T(\boldsymbol{p}_1)=T(\boldsymbol{p}_2)=Z$, where $Z$ equals one with
probability $1/2$ and zero otherwise.  The resulting path transmission has
expected value $\mathbb{E}[Z^2]=1/2$, whereas multiplying their individual
means gives $\mathbb{E}[Z]^2=1/4$.  Thus averaging the visibility values
computed from realized binary maps gives $1/2$, whereas evaluating the
recurrence once on the marginal-mean field gives $1/4$.  For this correlated
map model, only the former is the expected visibility; uniform selection within
the canonical path family neither prescribes nor removes the correlation
responsible for the difference.

In two dimensions, the single prefix count $K_{\ell}$ completely identifies
the cell reached on layer $\ell$.  A three-dimensional canonical path has two
transverse-advance schedules, so a layer cell is instead identified by two
prefix counts.  The multidimensional counterpart of
Eq.~\eqref{eq:hypergeometric-bridge-mass} is therefore the next question: what
joint distribution governs those counts, and why do the probabilities of the
four target-entry cases equal the bilinear predecessor coefficients?

\section{Three-Dimensional Product Bridge and Extension to \texorpdfstring{\(d\)}{d} Dimensions}
\label{sec:multidimensional-extension}

A 3D canonical path answers this question by carrying one transverse-advance
schedule for each of the two transverse coordinates.  Both schedules range
over the same $M$ layers and may select the same layer.  Their paired
choices define the 3D canonical path family, their prefix counts define the 3D
product bridge, and their choices on the final path step determine the
predecessor through which the path enters the target and thereby explain the
four bilinear coefficients.

\subsection{Paired schedules and the 3D product bridge}

Consider a 3D target
$\boldsymbol{q}=\boldsymbol{s}+h\boldsymbol{n}$ distinct from the source, and
retain the dominant-coordinate index $m$, layer count $M=|n_m|$, signs
$\sigma_j$, and ratios $p_j=|n_j|/M$ from
Section~\ref{sec:recurrence}.  Let
\begin{equation}
	\mathcal{R}_3=\{1,2,3\}\setminus\{m\}
	=\{r_1,r_2\},
	\qquad r_1<r_2,
	\label{eq:3d-transverse-index-set}
\end{equation}
be the set of transverse-coordinate indices labelled in coordinate order, and put
$R_j=|n_j|$ for $j\in\mathcal{R}_3$.  For each transverse coordinate, choose
a schedule
\begin{equation}
	S_j\subseteq\{1,\ldots,M\},
	\qquad |S_j|=R_j,
	\qquad j\in\mathcal{R}_3.
	\label{eq:3d-transverse-schedules}
\end{equation}
The schedule $S_j$ records the layers on which coordinate $j$ advances in
addition to the dominant coordinate.  An index may belong to neither schedule,
either schedule alone, or both schedules, so one layer step may change one,
two, or three coordinates.

Write $\boldsymbol{S}=(S_{r_1},S_{r_2})$ for a schedule pair.  After $\ell$
ordered path steps, define
\begin{equation}
	K_{j,\ell}
	=\left|S_j\cap\{1,\ldots,\ell\}\right|,
	\qquad j\in\mathcal{R}_3,
	\label{eq:3d-prefix-counts}
\end{equation}
so that the cell reached is
\begin{equation}
	\boldsymbol{q}_{\ell}(\boldsymbol{S})
	=\boldsymbol{s}+h\left(
	\ell\sigma_m\boldsymbol{e}_m
	+\sum_{j\in\mathcal{R}_3}
	K_{j,\ell}\sigma_j\boldsymbol{e}_j
	\right),
	\qquad 0\leq\ell\leq M.
	\label{eq:canonical-3d-path-point}
\end{equation}
As in 2D, the ordered sequence
$\pi_{\boldsymbol{S}}=(\boldsymbol{q}_1(\boldsymbol{S}),\ldots,
	\boldsymbol{q}_M(\boldsymbol{S}))$ forms a canonical path that excludes the
source and includes the target.  Schedule pairs and canonical paths correspond
one-to-one.  Each coordinate advances only in its source-to-target direction,
so every path remains inside the box spanned by its endpoints.  The 3D
canonical path family is therefore
\begin{equation}
	\begin{gathered}
		\mathcal{P}(\boldsymbol{s},\boldsymbol{q})
		=\left\{\pi_{\boldsymbol{S}}:
		S_j\subseteq\{1,\ldots,M\},\ |S_j|=R_j
		\text{ for }j\in\mathcal{R}_3\right\},\\
		|\mathcal{P}(\boldsymbol{s},\boldsymbol{q})|
		=\binom{M}{R_{r_1}}\binom{M}{R_{r_2}}.
	\end{gathered}
	\label{eq:canonical-3d-path-family}
\end{equation}
The product is essential.  It counts two choices made over the same $M$
layers, whereas a multinomial coefficient would assign each layer to one of
several mutually exclusive step types and would incorrectly forbid the two
transverse coordinates from advancing together.

Select a schedule pair uniformly from this Cartesian product.  The two
schedules are then selected independently, each conditioned on its required
endpoint count.  For
$\boldsymbol{k}=(k_1,k_2)$, the probability that the corresponding canonical
path reaches the cell with
$K_{r_1,\ell}=k_1$ and $K_{r_2,\ell}=k_2$ is
\begin{equation}
	\begin{gathered}
		H^{(3)}_{\ell,\boldsymbol{k}}
		:=\Pr\!\left(
		K_{r_1,\ell}=k_1,\ K_{r_2,\ell}=k_2\right) \\
		{}=\prod_{a=1}^{2}
		\frac{\binom{\ell}{k_a}
			\binom{M-\ell}{R_{r_a}-k_a}}
		{\binom{M}{R_{r_a}}}.
	\end{gathered}
	\label{eq:3d-product-bridge-mass}
\end{equation}
For each $a\in\{1,2\}$, the feasible count satisfies
\begin{equation}
	\max\{0,R_{r_a}-(M-\ell)\}
	\leq k_a\leq
	\min\{\ell,R_{r_a}\}.
	\label{eq:3d-product-bridge-support}
\end{equation}
The support is the Cartesian product of these two intervals.  Each factor in
Eq.~\eqref{eq:3d-product-bridge-mass} is the hypergeometric bridge derived in
Section~\ref{sec:hypergeometric-bridge}, so the product sums to one on every
layer.  We call this endpoint-conditioned joint distribution the \emph{3D
	product bridge}.  At $\ell=0$ it is concentrated at
$\boldsymbol{k}=(0,0)$, and at $\ell=M$ it is concentrated at
$\boldsymbol{k}=(R_{r_1},R_{r_2})$, giving endpoint pinching at the source and
target.  If $R_j$ is zero or $M$, the corresponding coordinate factor is a
point mass on every layer.

The product structure has two immediate geometric consequences.  For
$j\in\mathcal{R}_3$,
\begin{equation}
	\mathbb{E}[K_{j,\ell}]=\ell p_j,
	\qquad
	\operatorname{Cov}(K_{r_1,\ell},K_{r_2,\ell})=0,
	\label{eq:3d-product-bridge-moments}
\end{equation}
and mapping the mean counts through
Eq.~\eqref{eq:canonical-3d-path-point} gives
\begin{equation}
	\mathbb{E}\!\left[
		\boldsymbol{q}_{\ell}(\boldsymbol{S})\right]
	=\boldsymbol{s}+\frac{\ell}{M}
	(\boldsymbol{q}-\boldsymbol{s}).
	\label{eq:3d-product-bridge-mean-cell}
\end{equation}
For $M>1$, each coordinate variance is the 2D variance in
Eq.~\eqref{eq:hypergeometric-bridge-variance} with $R$ replaced by $R_j$.
When $M=1$, both variances are zero.  The coordinate variances quantify how the
cell visited by the selected path varies in the two transverse coordinates on
each intermediate layer.  The two transverse schedules are selected
independently because the schedule pair is selected uniformly from the
Cartesian product of the two admissible schedule sets.  This independence
concerns only canonical path selection and does not imply that cell
transmissions or occupancies in the map are independent.
Nor does the mean alignment imply exact agreement between the thresholded field
and continuous visibility at a fixed resolution.

\subsection{Target-entry probabilities and the bilinear update}

The same paired schedules explain the local 3D update.  For
$\boldsymbol{\epsilon}=(\epsilon_1,\epsilon_2)\in\{0,1\}^2$, let
$\epsilon_a=1$ mean that schedule $S_{r_a}$ selects step $M$.  Thus
$\boldsymbol{\epsilon}$ records whether neither transverse
coordinate, either coordinate alone, or both coordinates advance on the final
path step into the target.  The number of schedule pairs that produce a
specified target-entry case is
\begin{equation}
	\prod_{a=1}^{2}
	\binom{M-1}{R_{r_a}-\epsilon_a},
	\label{eq:3d-final-step-count}
\end{equation}
where a binomial coefficient is zero outside its natural range.  Dividing by
Eq.~\eqref{eq:canonical-3d-path-family} gives
\begin{equation}
	w_{\boldsymbol{\epsilon}}
	=\prod_{a=1}^{2}
	\omega_{r_a}(\epsilon_a),
	\qquad
	\omega_j(0)=1-p_j,
	\quad \omega_j(1)=p_j.
	\label{eq:3d-final-step-probability}
\end{equation}
These are probabilities within the canonical path family conditioned on its
fixed source and target.  They are not map probabilities or unconstrained
movement probabilities.

The predecessor for target-entry case $\boldsymbol{\epsilon}$ is
\begin{equation}
	\boldsymbol{q}_{\boldsymbol{\epsilon}}
	=\boldsymbol{q}-h\sigma_m\boldsymbol{e}_m
	-h\sum_{a=1}^{2}
	\epsilon_a\sigma_{r_a}\boldsymbol{e}_{r_a}.
	\label{eq:3d-final-step-predecessor}
\end{equation}
Consequently, the four cases and their probabilities are
\begin{center}
	\small
	\begin{tabular}{@{}ll@{}}
		\toprule
		\shortstack[l]{Transverse coordinates advancing\\into the target}
		& Path-family proportion   \\
		\midrule
		neither                                          & $(1-p_{r_1})(1-p_{r_2})$ \\
		$r_1$ only                                       & $p_{r_1}(1-p_{r_2})$     \\
		$r_2$ only                                       & $(1-p_{r_1})p_{r_2}$     \\
		both                                             & $p_{r_1}p_{r_2}$         \\
		\bottomrule
	\end{tabular}
\end{center}
These proportions are exactly the four weights in the bilinear interpolant of
Eq.~\eqref{eq:3d-interpolant}.  Figure~\ref{fig:predecessor-face-3d} shows the
correspondence for one asymmetric target.

\begin{figure*}[t]
	\centering
	\includegraphics[width=0.94\textwidth]{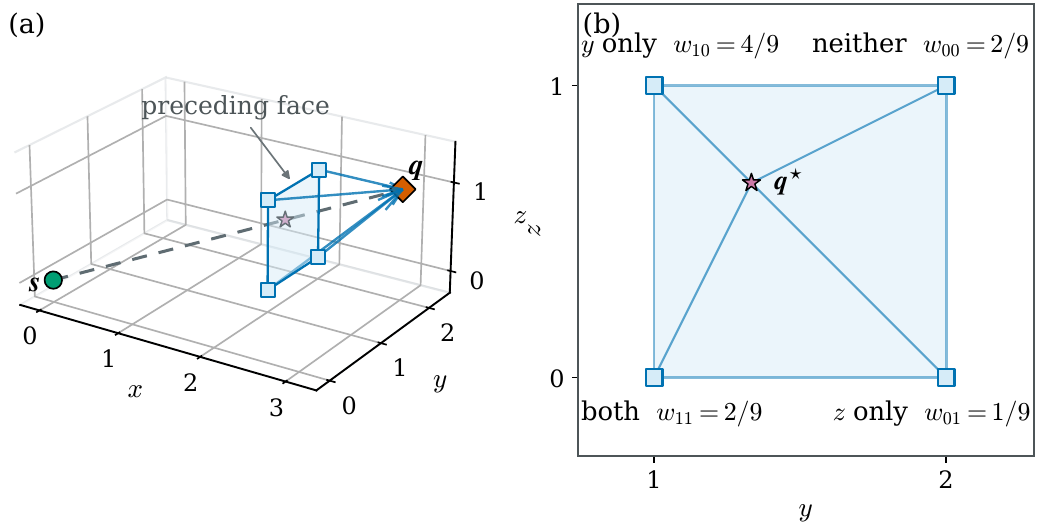}
	\caption{Four target-entry cases for the $x$-dominant target
		$\boldsymbol{q}=(3,2,1)$ with source $\boldsymbol{s}=(0,0,0)$ and
		$h=1$.  Every path enters the target from the preceding face $x=2$ on its
		final path step.  On that step, the path may additionally advance in
		neither transverse coordinate,
		$y$ only, $z$ only, or both.  The endpoint-conditioned proportions are
		obtained from $p_y=2/3$ and $p_z=1/3$ and equal, respectively,
		$2/9$, $4/9$, $1/9$, and $2/9$.  The dashed source--target segment meets
		the preceding face at the bilinear interpolation point
		$\boldsymbol{q}^{\star}=(2,4/3,2/3)$.}
	\label{fig:predecessor-face-3d}
\end{figure*}

To connect these local target-entry proportions to the full 3D canonical path
family in Eq.~\eqref{eq:canonical-3d-path-family}, define its finite-path
average by
\begin{equation}
	V_{\mathcal{P}}^{(3)}(\boldsymbol{q})
	=\frac{V_s}
	{\binom{M}{R_{r_1}}\binom{M}{R_{r_2}}}
	\sum_{\boldsymbol{S}}
	\mathcal{T}(\pi_{\boldsymbol{S}}),
	\qquad
	V_{\mathcal{P}}^{(3)}(\boldsymbol{s})=V_s,
	\label{eq:canonical-3d-path-average}
\end{equation}
where the sum ranges over the schedule pairs in
Eq.~\eqref{eq:3d-transverse-schedules}, and $\mathcal{T}$ is the transmission
product defined in Eq.~\eqref{eq:path-transmission-product}.

\begin{theorem}[Three-dimensional path-average identity]
	\label{thm:exact-3d-path-representation}
	\leavevmode\par\noindent
	For any source and target in a rectangular three-dimensional grid, any
	$V_s\in[0,1]$, and any deterministic transmission field
	$\tau:\mathcal{G}_h\to[0,1]$, the dominant-axis visibility recurrence in
	Eq.~\eqref{eq:visibility-recurrence} satisfies
	$V(\boldsymbol{q})=V_{\mathcal{P}}^{(3)}(\boldsymbol{q})$.
\end{theorem}

\begin{proof}
	For $M>0$, partition the schedule pairs in
	Eq.~\eqref{eq:canonical-3d-path-average} by
	$\boldsymbol{\epsilon}$.  Removing the target from every path in a target-entry
	case $\boldsymbol{\epsilon}$ with $w_{\boldsymbol{\epsilon}}>0$
	leaves the canonical path family ending at the predecessor
	$\boldsymbol{q}_{\boldsymbol{\epsilon}}$ in
	Eq.~\eqref{eq:3d-final-step-predecessor}.  The fraction of all canonical paths
	belonging to this target-entry case is
	$w_{\boldsymbol{\epsilon}}$ by
	Eqs.~\eqref{eq:3d-final-step-count}--\eqref{eq:3d-final-step-probability},
	and every path product contains the common target factor
	$\tau(\boldsymbol{q})$.  Hence
	\begin{equation}
		V_{\mathcal{P}}^{(3)}(\boldsymbol{q})
		=\tau(\boldsymbol{q})
		\sum_{\boldsymbol{\epsilon}\in\{0,1\}^2}
		w_{\boldsymbol{\epsilon}}
		V_{\mathcal{P}}^{(3)}
		(\boldsymbol{q}_{\boldsymbol{\epsilon}}).
		\label{eq:3d-path-average-recurrence}
	\end{equation}
	If several coordinates tie at magnitude $M$, every tied coordinate treated
	as transverse has $R_j=M$ and therefore advances on every layer.  Exchanging
	which tied coordinate is called dominant leaves each layer step, the path
	family, and the nonzero predecessor terms unchanged.  Every nonzero
	predecessor in Eq.~\eqref{eq:3d-path-average-recurrence} is consequently on
	layer $M-1$, even when the fixed coordinate-order tie rule selects a
	different dominant coordinate for the predecessor than for
	$\boldsymbol{q}$.  The source-cell base case is $V_s$, and induction over $M$
	proves the equality.
	\qed
\end{proof}

Theorem~\ref{thm:exact-3d-path-representation} establishes that the visibility
recurrence evaluates $V_s$ times the transmission-product
average over the same endpoint-fixed canonical path family that defines the 3D
product bridge.  The 3D product
bridge gives, on every layer after source initialization, the probability that
a uniformly selected source-to-target canonical path passes through each cell
in the bridge's support.
On the immediately preceding layer, the probability of each target-entry case
is exactly the corresponding bilinear coefficient.  The causal dynamic
program computes $V(\boldsymbol{q})$ through local reuse,
without enumerating the
$\binom{M}{R_{r_1}}\binom{M}{R_{r_2}}$ complete paths.  The finite-path
interpretation of the four-parent recurrence does not claim that
bilinear interpolation or a 3D causal visibility sweep is itself new.

The same path-average identity gives deterministic nonbinary transmissions a
concrete 3D meaning.  Consider a unit-spaced grid with source
$\boldsymbol{s}=(5,1,6)$ and target $\boldsymbol{q}=(1,3,3)$.  The corresponding
product bridge contains $24$ canonical paths, which can pass through the six
cells shown in Fig.~\ref{fig:nonbinary-layer-3d} on the intermediate layer
$x=3$.  Let $\mathcal{L}$ denote these six cells, and let every other
non-source cell be fully transmitting.  Each path visits exactly one cell in
$\mathcal{L}$.  For $V_s>0$, Eq.~\eqref{eq:heterogeneous-layer-response}
therefore gives
\[
	\frac{V(\boldsymbol{q})}{V_s}
	=\sum_{\boldsymbol{p}\in\mathcal{L}}
	H_{\boldsymbol{p}}\tau(\boldsymbol{p})
	=\frac{3}{4}.
\]
Thus $H_{\boldsymbol{p}}$ is exactly the weight of
$\tau(\boldsymbol{p})$ in the ratio $V(\boldsymbol{q})/V_s$, not a probability
that the map is free.

\begin{figure*}[t]
	\centering
	\includegraphics[width=\textwidth]{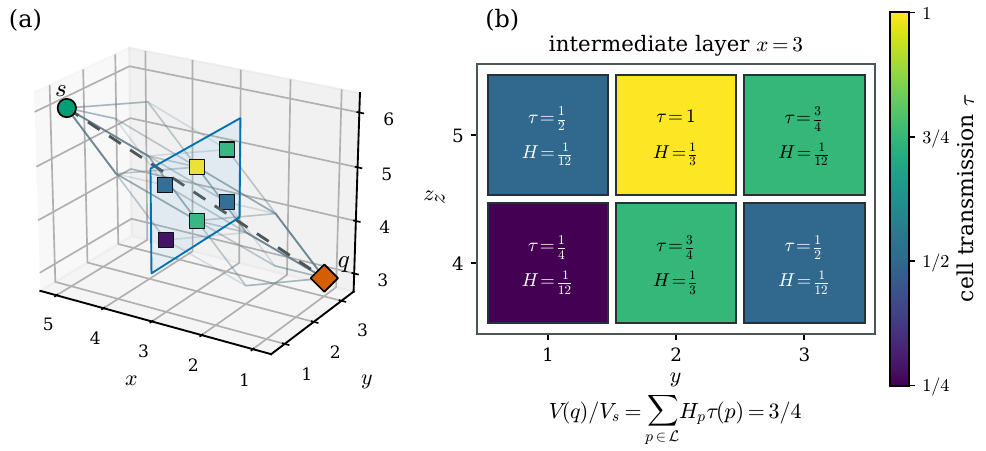}
	\caption{Bridge-weighted response of one partially transmitting 3D layer.
		(a) The product bridge contains $24$ canonical paths from
		$\boldsymbol{s}=(5,1,6)$ to $\boldsymbol{q}=(1,3,3)$.
		(b) At $x=3$, cell colour and the upper label give the deterministic
		transmission $\tau$, while the lower label gives the path-incidence
		probability $H$.  Every path visits exactly one of the six cells, so the
		incidence probabilities weight the layer transmissions to produce
		$V(\boldsymbol{q})/V_s=3/4$.}
	\label{fig:nonbinary-layer-3d}
\end{figure*}

Partially transmitting cells on different layers cannot generally be treated
separately because one path may visit several of them.  For the same
source--target family, let
$a=\tau(\boldsymbol{p})=1/4$ at $\boldsymbol{p}=(4,2,5)$ and
$b=\tau(\boldsymbol{r})=3/4$ at $\boldsymbol{r}=(2,3,4)$, with every other
non-source cell fully transmitting.  Writing the joint path-incidence
probability as
$H_{\boldsymbol{p},\boldsymbol{r}}
	=\Pr(\boldsymbol{p}\in\Pi\ \text{and}\ \boldsymbol{r}\in\Pi)$, where
$\Pi$ is selected uniformly from the $24$ paths, the exact incidences are
$H_{\boldsymbol{p}}=H_{\boldsymbol{r}}=3/8$ and
$H_{\boldsymbol{p},\boldsymbol{r}}=1/6$.  Averaging over the four possible
visit cases gives
\begin{equation}
	\frac{V(\boldsymbol{q})}{V_s}
	=1-(1-a)H_{\boldsymbol{p}}-(1-b)H_{\boldsymbol{r}}
	+(1-a)(1-b)H_{\boldsymbol{p},\boldsymbol{r}}
	=\frac{21}{32}.
	\label{eq:two-cell-joint-response}
\end{equation}
Treating the cells separately means evaluating two transmission fields, each
with only one partially transmitting cell.
Equation~\eqref{eq:single-cell-incidence-response} gives
$V(\boldsymbol{q})/V_s=23/32$ when that cell is $\boldsymbol{p}$ with
transmission $a=1/4$, and $V(\boldsymbol{q})/V_s=29/32$ when it is
$\boldsymbol{r}$ with transmission $b=3/4$.
Multiplying these two target responses would instead give
$(23/32)(29/32)=667/1024$.  That multiplication replaces the joint
path-incidence probability $1/6$ by
$H_{\boldsymbol{p}}H_{\boldsymbol{r}}=9/64$.  The mismatch therefore comes from
treating the two path-incidence events as independent, not from map uncertainty:
the transmission field is deterministic, and the probability refers only to
uniform selection from the fixed canonical path family.  The recurrence
therefore gives the exact average of the complete
path-transmission products over the declared canonical path family, including
the joint contribution $H_{\boldsymbol{p},\boldsymbol{r}}$ from paths that visit
both partially transmitting cells.  This average does not assume that the two
visit events are independent, and it is neither a general occupancy probability
nor a model of optical effects beyond the prescribed cell attenuation.  The
next subsection extends the finite-path average and bridge construction to
any fixed spatial dimension.

\subsection{Product construction in \texorpdfstring{\(d\)}{d} dimensions}

The 3D argument contains no dimension-specific step beyond the number of
transverse schedules.  For any $d\geq1$, consider a target
$\boldsymbol{q}=\boldsymbol{s}+h\boldsymbol{n}$ distinct from the source and
extend the source-relative definitions of Section~\ref{sec:recurrence} to $d$
coordinates.  Choose the smallest coordinate index $m$ satisfying $|n_m|=M$,
and define the transverse-coordinate set
\begin{equation}
	\mathcal{R}_d=\{1,\ldots,d\}\setminus\{m\}.
	\label{eq:fixed-d-transverse-index-set}
\end{equation}
For every $j\in\mathcal{R}_d$, independently and uniformly select a schedule
$S_j\subseteq\{1,\ldots,M\}$ of size $R_j=|n_j|$.  Write
$\boldsymbol{S}=(S_j)_{j\in\mathcal{R}_d}$ and
$K_{j,\ell}=|S_j\cap\{1,\ldots,\ell\}|$.  The cell reached on layer $\ell$ is
\begin{equation}
	\boldsymbol{q}_{\ell}(\boldsymbol{S})
	=\boldsymbol{s}+h\left(
	\ell\sigma_m\boldsymbol{e}_m
	+\sum_{j\in\mathcal{R}_d}
	K_{j,\ell}\sigma_j\boldsymbol{e}_j
	\right),
	\qquad 0\leq\ell\leq M.
	\label{eq:fixed-d-path-point}
\end{equation}
The $M$ cells reached for $1\leq\ell\leq M$ form the canonical path
$\pi_{\boldsymbol{S}}$.  The path uniquely recovers its schedule tuple.  The
cardinality of the \(d\)-dimensional canonical path family and the
product-bridge probability at a prefix-count vector
$\boldsymbol{k}=(k_j)_{j\in\mathcal{R}_d}$ are, respectively,
\begin{equation}
	|\mathcal{P}_d(\boldsymbol{s},\boldsymbol{q})|
	=\prod_{j\in\mathcal{R}_d}\binom{M}{R_j},
	\qquad
	H^{(d)}_{\ell,\boldsymbol{k}}
	=\prod_{j\in\mathcal{R}_d}
	\frac{\binom{\ell}{k_j}
		\binom{M-\ell}{R_j-k_j}}
	{\binom{M}{R_j}}.
	\label{eq:fixed-d-cardinality-and-bridge}
\end{equation}
For each $j\in\mathcal{R}_d$, the feasible prefix count satisfies
\[
	\max\{0,R_j-(M-\ell)\}
	\leq k_j\leq
	\min\{\ell,R_j\}.
\]
The support of $H^{(d)}_{\ell,\boldsymbol{k}}$ is the Cartesian product of
these intervals, and $H^{(d)}_{\ell,\boldsymbol{k}}=0$ outside that support.
Empty products equal one, so $d=1$ gives the unique signed axis path.  For a
target-entry vector
$\boldsymbol{\epsilon}=(\epsilon_j)_{j\in\mathcal{R}_d}
	\in\{0,1\}^{d-1}$, which records the transverse coordinates that advance into
the target, the corresponding probability and predecessor are
\begin{equation}
	\begin{gathered}
		w_{\boldsymbol{\epsilon}}
		=\prod_{j\in\mathcal{R}_d}
		\omega_j(\epsilon_j),                        \\
		\boldsymbol{q}_{\boldsymbol{\epsilon}}
		=\boldsymbol{q}-h\sigma_m\boldsymbol{e}_m
		-h\sum_{j\in\mathcal{R}_d}
		\epsilon_j\sigma_j\boldsymbol{e}_j.
	\end{gathered}
	\label{eq:fixed-d-final-step-law}
\end{equation}

Define the \(d\)-dimensional finite-path average by
\begin{equation}
	V_{\mathcal{P}}^{(d)}(\boldsymbol{q})
	=\frac{V_s}{|\mathcal{P}_d(\boldsymbol{s},\boldsymbol{q})|}
	\sum_{\boldsymbol{S}}\mathcal{T}(\pi_{\boldsymbol{S}}),
	\qquad
	V_{\mathcal{P}}^{(d)}(\boldsymbol{s})=V_s,
	\label{eq:fixed-d-path-average}
\end{equation}
where the sum ranges over all schedule tuples above.  The value at a
coincident source and target is therefore supplied by the single empty path.

\begin{theorem}[Path-average identity in \(d\) dimensions]
	\label{thm:fixed-d-product-bridge}
	\leavevmode\par\noindent
	For any $d\geq1$, any source and target in a $d$-dimensional rectangular
	grid, any $V_s\in[0,1]$, and any deterministic transmission field
	$\tau$ taking values in $[0,1]$ on that grid, the probabilities in the
	product bridge of Eq.~\eqref{eq:fixed-d-cardinality-and-bridge} sum to one on
	every layer when $\boldsymbol{q}\ne\boldsymbol{s}$.  With
	$V(\boldsymbol{s})=V_s$, the
	\(d\)-dimensional dominant-axis recurrence
	\begin{equation}
		V(\boldsymbol{q})
		=\tau(\boldsymbol{q})
		\sum_{\boldsymbol{\epsilon}\in\{0,1\}^{d-1}}
		w_{\boldsymbol{\epsilon}}
		V(\boldsymbol{q}_{\boldsymbol{\epsilon}})
		\qquad(\boldsymbol{q}\ne\boldsymbol{s})
		\label{eq:fixed-d-visibility-recurrence}
	\end{equation}
	satisfies
	$V(\boldsymbol{q})=V_{\mathcal{P}}^{(d)}(\boldsymbol{q})$
	and $0\leq V(\boldsymbol{q})\leq1$.
	The computed visibility value and the nonzero predecessors with their weights
	are invariant under the choice of a tied dominant coordinate.  If
	\begin{equation}
		\nu
		=\bigl|\{j\in\mathcal{R}_d:0<p_j<1\}\bigr|,
		\label{eq:fractional-transverse-count}
	\end{equation}
	then exactly $2^\nu$ of the $2^{d-1}$ nominal choices give distinct
	predecessors with nonzero weight.
\end{theorem}

\begin{proof}
	Uniform selection from the Cartesian product of schedule families gives the
	product bridge in Eq.~\eqref{eq:fixed-d-cardinality-and-bridge}.
	The product-bridge probabilities sum to one because the probabilities in each
	one-coordinate factor sum to one.
	Partitioning the schedule tuples by their target-entry vector
	$\boldsymbol{\epsilon}$ gives
	Eq.~\eqref{eq:fixed-d-final-step-law}.  Removing the target maps the paths in each
	target-entry case $\boldsymbol{\epsilon}$ with $w_{\boldsymbol{\epsilon}}>0$
	bijectively to the canonical path family of its predecessor.  Factoring the
	common target transmission out of every path product, as in
	Theorem~\ref{thm:exact-3d-path-representation}, shows that
	$V_{\mathcal{P}}^{(d)}$ obeys
	Eq.~\eqref{eq:fixed-d-visibility-recurrence} with every occurrence of $V$
	replaced by $V_{\mathcal{P}}^{(d)}$.

	If $p_j=0$, nonzero weight forces $\epsilon_j=0$.  If $p_j=1$, it forces
	$\epsilon_j=1$.  If $0<p_j<1$, both choices have positive weight.  The
	$\nu$ transverse ratios satisfying $0<p_j<1$ therefore produce exactly $2^\nu$
	distinct nonzero predecessors.  Every tied transverse coordinate has
	$p_j=1$ and is forced to advance on each layer, so exchanging which tied
	coordinate is designated dominant leaves the canonical path family and every
	nonzero predecessor with its weight unchanged.  Because all tied coordinates
	are forced to advance, every
	nonzero predecessor lies on layer $M-1$ and inside the source--target box.
	Induction over $M$, starting
	from the common source-cell base case $V_s$, therefore proves
	$V=V_{\mathcal{P}}^{(d)}$.  Finally, the nonnegative weights sum to one, so
	$0\leq\tau,V_s\leq1$ implies $0\leq V\leq1$.
	\qed
\end{proof}

For $d=2$ and $d=3$, the \(d\)-dimensional path-average identity reduces,
respectively, to the two-parent and four-parent identities in
Theorems~\ref{thm:exact-2d-path-representation}
and~\ref{thm:exact-3d-path-representation}.  The \(d\)-dimensional product
construction therefore
yields a single path-average identity rather than separate formulas for each
dimension.  The result supplies a finite-path meaning for the established
multilinear stencil.  It does not claim a new interpolation rule.  For a
complete field with $N_v$ returned grid values,
Eq.~\eqref{eq:fixed-d-visibility-recurrence} uses at most $2^{d-1}$ local terms
per value and therefore requires $O(2^{d-1}N_v)$ work.  Writing the returned
field already costs $\Omega(N_v)$ operations, so the work is $\Theta(N_v)$
when $d$ is fixed, with an exponential dimension-dependent stencil constant.

The exact path-average identities do not by themselves establish agreement
with continuous visibility under grid refinement.  To relate the two, we
use the bridge law to show that a uniformly selected canonical path
concentrates around the straight source--target segment as the grid spacing
decreases.

\subsection{Agreement with continuous visibility under grid refinement}
\label{sec:continuous-visibility-refinement}

Fix a rectangular physical domain $\mathcal{D}\subset\mathbb{R}^d$, a closed
obstacle set $\mathcal{O}\subset\mathcal{D}$, and distinct physical endpoints
$\bar{\boldsymbol{s}},\bar{\boldsymbol{q}}\in
\mathcal{D}\setminus\mathcal{O}$.  Consider a nested family of cubic grids
with spacing $h\to0$.  Let $\boldsymbol{s}_h$ and
$\boldsymbol{q}_h$ be the centres of the cells containing the two physical
endpoints.  If an endpoint lies on a shared cell boundary, select one of the
adjacent cells by a fixed rule.  Then
\[
	\|\boldsymbol{s}_h-\bar{\boldsymbol{s}}\|_\infty\leq\frac{h}{2},
	\qquad
	\|\boldsymbol{q}_h-\bar{\boldsymbol{q}}\|_\infty\leq\frac{h}{2}.
\]
A grid cell is blocked when its closed cell box intersects
$\mathcal{O}$.  Accordingly, the binary transmission field $\tau_h$ is zero
on blocked cells and one elsewhere.  Let
$V_h(\boldsymbol{q}_h)$ denote the scalar visibility value defined by the
recurrence in exact arithmetic, with prescribed source value
$V_h(\boldsymbol{s}_h)=1$ and no additional decay.

Write
$\boldsymbol{\gamma}(t)
=(1-t)\bar{\boldsymbol{s}}+t\bar{\boldsymbol{q}}$ for
$0\leq t\leq1$.  The physical query is continuously visible when the open
segment $\boldsymbol{\gamma}((0,1))$ does not intersect
$\mathcal{O}$.  Contact with the obstacle boundary counts as blocked.
For a fixed threshold $\theta\in(0,1)$, the recurrence instead classifies the
sampled target as visible when
$V_h(\boldsymbol{q}_h)\geq\theta$.

The refinement result covers two cases.  The segment has \emph{positive
clearance} when its $L^\infty$ distance from $\mathcal{O}$ is positive.
Alternatively, the segment enters the obstacle interior: for some
$t_0\in(0,1)$, an $L^\infty$ ball of positive radius around
$\boldsymbol{\gamma}(t_0)$ lies inside $\mathcal{O}$.  We call queries
satisfying either condition \emph{nongrazing}.  A segment that touches an
obstacle boundary but never enters any obstacle interior satisfies neither
condition.

\begin{theorem}[Nongrazing agreement under refinement]
	\label{thm:nongrazing-refinement-consistency}
	\leavevmode\par\noindent
	In any fixed dimension $d\geq1$, consider the refinement setting above.
	For all sufficiently small $h$, the sampled source and target cells are
	free.  If the physical source--target segment has positive clearance, then
	\[
		V_h(\boldsymbol{q}_h)\longrightarrow1
		\qquad\text{as }h\to0.
	\]
	If the segment enters the obstacle interior, then
	\[
		V_h(\boldsymbol{q}_h)\longrightarrow0
		\qquad\text{as }h\to0.
	\]
	Consequently, for every fixed threshold $\theta\in(0,1)$, the thresholded
	recurrence classification agrees with continuous visibility on every
	sufficiently fine grid in either case.
\end{theorem}

\begin{proof}
	Because $\tau_h$ is binary and the prescribed source value is one,
	Theorem~\ref{thm:fixed-d-product-bridge} identifies
	$V_h(\boldsymbol{q}_h)$ with the probability that a uniformly selected
	canonical path encounters no blocked cell.

	Write
	\[
		\boldsymbol{q}_h-\boldsymbol{s}_h
		=h\boldsymbol{n}_h,
		\qquad
		M_h=\|\boldsymbol{n}_h\|_\infty,
	\]
	so that $M_h\geq1$ for all sufficiently small $h$, and define
	$\boldsymbol{\gamma}_h(t)=(1-t)\boldsymbol{s}_h+t\boldsymbol{q}_h$.
	For a uniformly selected
	canonical path, let $\boldsymbol{X}_{h,\ell}$ be the centre of the cell
	reached after $\ell$ path steps.  In each transverse coordinate $j$, let
	$K_{j,\ell}$ count the advances completed by layer $\ell$.  This count
	follows the hypergeometric bridge established above, with mean
	$\ell|n_{h,j}|/M_h$.  Consequently,
	$\mathbb{E}[\boldsymbol{X}_{h,\ell}]
	=\boldsymbol{\gamma}_h(\ell/M_h)$.

	For $1\leq\ell\leq M_h$ and each transverse coordinate, Hoeffding's
	inequality for sampling without replacement
	\cite{Hoeffding1963ProbabilityInequalities} gives
	\[
		\Pr\!\left(
		\left|K_{j,\ell}-\frac{\ell|n_{h,j}|}{M_h}\right|
		\geq\frac{r}{h}\right)
		\leq 2\exp\!\left(-\frac{2r^2}{h^2\ell}\right)
		\leq 2\exp\!\left(-\frac{2r^2}{h^2M_h}\right).
	\]
	The endpoints are deterministic.  A union bound over the transverse
	coordinates and layers therefore gives, for every fixed $r>0$,
	\begin{equation}
		\Pr\!\left(
			\max_{0\leq\ell\leq M_h}
			\left\|
				\boldsymbol{X}_{h,\ell}
				-\boldsymbol{\gamma}_h(\ell/M_h)
			\right\|_\infty
			\geq r
		\right)
		\leq
		2(d-1)M_h
		\exp\!\left(-\frac{2r^2}{h^2M_h}\right).
		\label{eq:canonical-path-concentration}
	\end{equation}
	Coordinates tied for dominance advance on every layer and introduce no
	additional deviation.  For $d=1$, the path is deterministic and the
	deviation probability is zero.  Since the physical endpoints are distinct,
	$hM_h\to
	\|\bar{\boldsymbol{q}}-\bar{\boldsymbol{s}}\|_\infty>0$.
	The right-hand side of
	Eq.~\eqref{eq:canonical-path-concentration} therefore tends to zero.
	Thus a uniformly selected canonical path remains inside every fixed-width
	tube around the sampled straight segment with probability approaching one.

	Since $\mathcal{O}$ is closed and the two physical endpoints lie outside
	it, each endpoint has positive distance from $\mathcal{O}$.  Their sampled
	cells are therefore free for all sufficiently small $h$.

	Suppose first that the physical segment has positive clearance
	\[
		\delta
		=\operatorname{dist}_\infty
		\bigl(\boldsymbol{\gamma}([0,1]),\mathcal{O}\bigr)>0.
	\]
	The endpoint-sampling bounds imply
	\[
		\max_{0\leq t\leq1}
		\|\boldsymbol{\gamma}_h(t)-\boldsymbol{\gamma}(t)\|_\infty
		\leq\frac{h}{2}.
	\]
	Choose $0<r<\delta/3$.  A cell extends at most $h/2$ from its centre in
	$L^\infty$ distance.  Thus every point of a cell whose centre lies within
	distance $r$ of the sampled segment is at most $r+h$ from the physical
	segment.  For sufficiently small $h$, we have $r+h<\delta$, so every such
	cell is free.  Equation~\eqref{eq:canonical-path-concentration} then gives
	$V_h(\boldsymbol{q}_h)\to1$.

	Suppose instead that the segment enters the obstacle interior.  Choose
	$t_0\in(0,1)$ and $\delta>0$ such that the $L^\infty$ ball of radius
	$\delta$ around $\boldsymbol{\gamma}(t_0)$ lies inside
	$\mathcal{O}$.  Let $\ell_h$ be a layer index nearest to $t_0M_h$.
	Then
	\[
		\boldsymbol{\gamma}_h(\ell_h/M_h)
		\longrightarrow\boldsymbol{\gamma}(t_0).
	\]
	Applying Eq.~\eqref{eq:canonical-path-concentration} with
	$0<r<\delta/3$ shows that the cell centre reached on layer $\ell_h$ lies
	inside that ball with probability approaching one.
	That cell is blocked, so the probability that a canonical path remains
	clear tends to zero.  Therefore $V_h(\boldsymbol{q}_h)\to0$.

	Since $\theta\in(0,1)$ is fixed, these limits imply
	$V_h(\boldsymbol{q}_h)\geq\theta$ in the clear case and
	$V_h(\boldsymbol{q}_h)<\theta$ in the blocked case for all sufficiently
	small $h$.
	\qed
\end{proof}

The theorem guarantees eventual agreement for each fixed nongrazing query,
not exact classification at a prescribed grid spacing.  It does not cover
segments that meet an obstacle boundary without entering any obstacle
interior.  Nor does it select a preferred threshold or provide a convergence
rate uniform over all queries.  Finite-resolution errors remain a separate
validation question.

\section{Numerical Methods and Implementations}
\label{sec:implementations}

The recurrence in Eq.~\eqref{eq:fixed-d-visibility-recurrence} defines each
non-source visibility value locally.  A full-field computation must order those
local dependencies and evaluate them once for every grid cell.  This section
makes that dynamic program explicit, separates exact-arithmetic verification
from full-volume computation, and shows how the four-parent NumPy, OpenMP, and
CUDA implementations preserve the same 3D recurrence.
The main consequence is computational: the finite-path average can be
evaluated without enumerating every member of the canonical path family, whose
size can grow exponentially with the source--target grid distance, while
the causal layers expose parallel work without changing the mathematical
model.  The section closes by specifying the transmission field and initial
visibility value used by each implementation and the scalar field each returns,
checking numerical agreement among the NumPy, OpenMP, and CUDA implementations,
and introducing the three alternative methods used to compare classification
error, boundary displacement, and CPU throughput in
Section~\ref{sec:validation}.

\subsection{Causal full-field algorithm}
\label{sec:causal-full-field-algorithm}

Consider a rectangular grid \(\mathcal{G}_h\) with \(N_v\) cells, transmission
field \(\tau\), source \(\boldsymbol{s}\), and initial visibility value \(V_s\).
A \emph{full-field computation} returns \(V(\boldsymbol{q})\) for every
\(\boldsymbol{q}\in\mathcal{G}_h\) for that fixed source.  Define the layer index
of a cell by
\[
	L(\boldsymbol{q})
	=\frac{\|\boldsymbol{q}-\boldsymbol{s}\|_{\infty}}{h},
\]
which is the largest source-relative coordinate displacement measured in grid
steps, and let \(L_{\max}=\max_{\boldsymbol{q}\in\mathcal{G}_h}
L(\boldsymbol{q})\).  The fixed coordinate order used below selects the first
maximum, namely \(x\), then \(y\), then \(z\) in 3D.  This fixed tie-breaking
order gives each cell one update even when two or more displacements are tied.

Algorithm~\ref{alg:causal-full-field} is the complete source-to-all-cells
calculation.  The visualizations in
Figure~\ref{fig:visibility-3d-showcase} are derived from fields returned by
this calculation.  It invokes the predecessors and weights already defined in
Eq.~\eqref{eq:fixed-d-final-step-law}; it does not introduce another local
visibility rule.

\begin{algorithm*}[t]
	\caption{Causal evaluation of the dominant-axis visibility field}
	\label{alg:causal-full-field}
	\begin{algorithmic}[1]
		\Require rectangular grid \(\mathcal{G}_h\), transmission \(\tau\), source
		\(\boldsymbol{s}\), and initial visibility value \(V_s\)
		\Ensure visibility field \(V\) on \(\mathcal{G}_h\)
		\State \(V(\boldsymbol{q})\gets0\) for every
		\(\boldsymbol{q}\in\mathcal{G}_h\); \(V(\boldsymbol{s})\gets V_s\)
		\For{\(\ell=1,\ldots,L_{\max}\)}
		\ForAll{\(\boldsymbol{q}\in\mathcal{G}_h\) with
			\(L(\boldsymbol{q})=\ell\)}
		\State \(m\gets\min\!\left\{j\in\{1,\ldots,d\}:
		|q_j-s_j|=\|\boldsymbol{q}-\boldsymbol{s}\|_\infty\right\}\)
		\State obtain the nonzero
		\((\boldsymbol{q}_{\boldsymbol{\epsilon}},
		w_{\boldsymbol{\epsilon}})\) pairs from
		Eq.~\eqref{eq:fixed-d-final-step-law} using \(m\)
		\State \(V(\boldsymbol{q})\gets\tau(\boldsymbol{q})
		\displaystyle\sum_{w_{\boldsymbol{\epsilon}}>0}
		w_{\boldsymbol{\epsilon}}
		V(\boldsymbol{q}_{\boldsymbol{\epsilon}})\)
		\EndFor
		\EndFor
		\State \Return \(V\)
	\end{algorithmic}
\end{algorithm*}

Every nonzero predecessor of a layer-\(\ell\) cell lies on layer
\(\ell-1\), as established in Theorem~\ref{thm:fixed-d-product-bridge}.
Consequently, all required values are available before a layer begins, cells
within that layer have no mutual dependencies, and each cell writes only its
own visibility value.  A synchronization point between consecutive layers is
therefore sufficient for a parallel implementation.  The ordering is also valid
on axes and tied diagonals because only nonzero-weight predecessors enter the
update.

\subsection{OpenMP and CUDA full-volume implementations}
\label{sec:parallel-full-volume-realizations}

In 3D, the OpenMP and CUDA implementations follow
Algorithm~\ref{alg:causal-full-field}, processing the layers in increasing
\(L\), outward from the source.  The cells on each layer form the boundary of
an axis-aligned box centred at the source and clipped by the volume boundary.
The OpenMP and CUDA implementations both divide that boundary into at most six
disjoint faces.  The two \(x\)-faces retain all their edges and corners.  The
\(y\)-faces omit cells
already owned by an \(x\)-face, and the \(z\)-faces omit cells already owned by either an
\(x\)- or \(y\)-face.  The resulting face assignment applies the fixed
\(x\), \(y\), \(z\) tie order: every layer cell belongs to
exactly one face, including axis-aligned, face-diagonal, and space-diagonal
cells.

The OpenMP implementation uses a persistent-shell schedule, which keeps one set
of CPU threads active throughout the complete outward sweep.  On each layer,
the disjoint face partition prevents cells from being assigned twice, while a
fixed work assignment distributes each face's cells among the available
threads.
Each cell is written by one thread, which accumulates up to four predecessor
contributions in a fixed local order.  Because every predecessor lies on the
completed preceding layer, the threads can evaluate the current layer
independently, and the next outward layer begins only after all of its cells
have been computed.  CPU parallelism therefore changes who evaluates a cell,
not the four-parent update or outward layer order that determines its
visibility value.

The CUDA implementation uses a fused schedule that combines all owned faces of
an outward layer in one kernel launch and assigns one GPU thread to each face
cell.  Successive kernel launches follow the outward order, ensuring that every
preceding layer is complete before its values are read.  The transmission and visibility fields
remain in GPU memory throughout a solve.  The CPU and GPU implementations
therefore differ in how they assign independent cells to hardware, not in the
recurrence they evaluate.

Algorithm~\ref{alg:causal-full-field} visits each non-source cell once.  When
\(d\) is fixed, each visit applies a bounded local update, so the work is
\(\Theta(N_v)\), while the transmission and visibility fields require
\(\Theta(N_v)\) storage.  The sequential dependency depth is \(L_{\max}\) in
every dimension, although the cells within each layer may be evaluated in
parallel.  Once the transmission grid has been supplied, the work scales with
\(N_v\), the number of cells in the returned visibility field, rather than with
the number or polygonal complexity of obstacle boundaries.  These are
algorithmic cost and storage results; Section~\ref{sec:validation} reports
measured runtime, throughput, host and device memory, and transfer costs under
declared experimental conditions.

\subsection{Recurrence verification and numerical agreement}
\label{sec:numerical-interface-agreement}

The software for the dominant-axis recurrence serves two distinct purposes.
Exact-arithmetic evaluators compute the recurrence in 2D, 3D, and any fixed
spatial dimension.  These evaluators are correctness tools, not performance
implementations.  The four-parent 3D
recurrence also has three single-precision full-volume implementations: the
NumPy implementation exposes the layer grouping and predecessor operations
directly, the OpenMP implementation provides optimized CPU computation, and
the CUDA implementation provides optimized GPU computation.

The four-parent OpenMP and CUDA implementations exchange rectangular 3D fields
as contiguous single-precision arrays, with consecutive \(z\)-indices stored
adjacently, followed by \(y\) and then \(x\).  All reported experiments use the
unit-source convention \(V_s=1\) and no additional decay.  The NumPy and OpenMP
implementations also support the transmitting-source convention defined in
Section~\ref{sec:recurrence}.  An optional uniform per-layer decay factor
\(\alpha\in[0,1]\) is incorporated by replacing
\(\tau(\boldsymbol{q})\) with \(\alpha\tau(\boldsymbol{q})\) in
Eq.~\eqref{eq:visibility-recurrence} and in the non-source update of
Algorithm~\ref{alg:causal-full-field}.  When \(\tau\equiv1\), this gives
\(V(\boldsymbol{q})=V_s\alpha^{L(\boldsymbol{q})}\).
The OpenMP and CUDA implementations return the scalar visibility field.
Thresholding is applied only after that field has been returned and is not part
of the recurrence.

The implementations are checked against the four-parent recurrence in two
stages.  On test cases with tractable path counts, visibility values returned
by the exact-arithmetic evaluators match the exhaustive average over the
canonical path family.  The
NumPy, OpenMP, and CUDA implementations then compute complete single-precision
fields independently.  Their floating-point operations may be ordered or
combined differently, so agreement between implementations is numerical
rather than necessarily bitwise.

For the six test scenes G0--G5 sampled on \(64^3\) grids, with cell transmissions
restricted to zero or one, the maximum absolute OpenMP--CUDA field difference
was \(5.96\times10^{-7}\).  For a separate \(32^3\) grid containing partially
transmitting cells, the maximum absolute difference between the visibility
fields computed by NumPy, OpenMP, and CUDA was below \(2.99\times10^{-7}\).
For these binary and partially transmitting inputs, the
OpenMP fields were bitwise identical across the tested thread counts and
repeated runs, while repeated CUDA runs also returned bitwise-identical fields.
Across the CUDA schedules tested in the larger study described in
Section~\ref{sec:performance-results}, the maximum absolute OpenMP--CUDA
difference was \(4.11\times10^{-6}\) for
\(N\in\{64,128,256,512\}\).  The fused CUDA comparison at \(N=1000\)
reached \(2.56\times10^{-6}\).  All comparisons passed their declared
numerical tolerances.

No implementation clamps the computed values.  The largest nonbinary value was
\(1+1.19\times10^{-7}\), which exceeds the exact upper bound \(V\leq1\)
established in Theorem~\ref{thm:fixed-d-product-bridge} only by an amount
consistent with single-precision roundoff.  The scalar fields are therefore
the primary agreement quantity: threshold labels may differ when roundoff
places two implementations' values on opposite sides of a declared threshold.

The accompanying repository provides compact reproducibility tests that check
the exact-arithmetic 2D recurrence against exhaustive finite-path enumeration
and compare the four-parent NumPy, OpenMP, and CUDA fields in 3D.  Together,
agreement with exhaustive canonical-path enumeration and bounded cross-implementation field
differences establish that the NumPy, OpenMP, and CUDA implementations preserve
the declared recurrence up to the reported single-precision differences.  This
implementation validation establishes that the three full-volume
implementations compute the intended four-parent field, but does not establish
that field's geometric accuracy or its performance relative to other methods.

\subsection{Comparison methods}
\label{sec:comparison-methods}

The evaluation compares the four-parent method with three alternative
grid-visibility methods.  Separate comparison software implements two methods
described by Goldstein et al.~\cite{Goldstein2022PathCounting}: an edge-gated
six-neighbour axial method and an edge-gated 26-neighbour tetrahedral method.
It also implements a controlled tetrahedral variant that keeps the tetrahedral
predecessor family but replaces each predecessor--target edge-visibility test
with a check that the destination cell is free.
Comparing this controlled variant with the edge-gated tetrahedral method
isolates the effect of predecessor--target edge-visibility tests while keeping
the tetrahedral predecessor family unchanged.

All four methods receive the same binary transmission field and source and
return scalar visibility fields under the unit-source convention, with no
additional decay.  The three alternatives differ from the four-parent
recurrence in their predecessor sets and local blocking tests, so agreement
with the four-parent field is neither expected nor used as an implementation
check.

\section{Numerical Validation and Computational Performance}
\label{sec:validation}

Recurrence correctness alone establishes neither agreement with obstacle
geometry nor computational cost.  We therefore evaluate the four-parent method
and the three comparison methods defined in
Section~\ref{sec:comparison-methods}.

The geometric evaluation begins with a prescribed source position and
continuous obstacle geometry.  Sampling this physical scene on a grid places
the source at a cell centre, which we call the \emph{grid source}, and marks
each obstacle-intersecting cell as blocked.  The resulting binary transmission
field and grid source are
supplied to all four methods.  Each method returns a scalar visibility field,
which a chosen threshold converts into a binary prediction.  The evaluation
compares that prediction with visibility in both the physical scene and the
binary input grid.

Subsection~\ref{sec:validation-design} defines the two comparison labels and the
thresholded method prediction, then specifies the six test scenes.
Subsection~\ref{sec:geometry-refinement-results} reports classification error
and boundary displacement under grid refinement, together with limitations
arising from threshold choice and near-contact sampling.
Subsection~\ref{sec:performance-results} compares the four methods' CPU
performance and reports schedule-dependent throughput and memory use.

\subsection{Visibility labels and test scenes}
\label{sec:validation-design}

All six scenes use the coordinate domain $[0,1]^3$, a cube of unit side length.
The prescribed source position and obstacle geometry remain fixed as the grid
is refined.  At resolution $N$, each coordinate is divided into $N$ cells of
width $h=1/N$.  The grid source is the centre of the cell containing the
prescribed source position.  A cell is blocked when the cell's closed cube
intersects an obstacle.  Blocked cells receive $\tau=0$, while all other cells
receive $\tau=1$.  Every method receives the same binary transmission field,
the same source cell, and no additional decay under the unit-source convention.

For each target in a free cell, the physical scene and grid input provide
separate binary visibility labels.  The \emph{continuous-visibility label}
tests the physical scene and marks a target as visible when the open straight
segment from the grid source $\boldsymbol{s}$ to the target centre
$\boldsymbol{q}$,
\[
	\bigl\{(1-t)\boldsymbol{s}+t\boldsymbol{q}:0<t<1\bigr\},
\]
does not intersect an obstacle.  The restriction $0<t<1$ excludes the source
and target endpoints.  Contact with an obstacle boundary counts as blocked.
The \emph{blocked-cell label} instead tests the binary input grid and marks the
target as visible when every non-source cell touched by the closed
source--target segment is free.  A disagreement between the two labels records
a difference between the physical scene and its blocked-cell representation
before the recurrence is evaluated.  Neither label specifies a numerical value
for $V(\boldsymbol{q})$.

Thresholding the returned scalar field produces a third binary classification:
a target is predicted visible when $V(\boldsymbol{q})\geq\theta$.  We report
five thresholds,
\[
	\theta\in\{1/10,1/4,1/2,3/4,9/10\},
\]
without selecting a preferred value.
Relative to the continuous-visibility label, the false-negative rate is the
fraction of continuously visible free targets predicted invisible, while the
false-positive rate is the fraction of continuously blocked free targets
predicted visible.

The blocked-cell label separates each error relative to continuous visibility
into two cases.  If a continuously visible target is predicted invisible and
the blocked-cell label also marks it as blocked, the prediction agrees with the
binary grid supplied to the solver.  If the blocked-cell label instead marks
the target as visible, the prediction disagrees with both comparison labels.
We call the second case a \emph{residual false negative}.  The same partition
applies to false positives.  The resulting partition identifies which
classifications coincide but does not establish what caused a disagreement.

The six scenes serve three experimental roles.  G0 provides the free-space base
case.  G1--G4 provide full-field refinement cases at
$N\in\{32,64,128,256\}$: an opaque box (G1), a thin obstacle (G2), a wall with a
narrow aperture (G3), and an obliquely viewed box (G4).  G5 tests near-contact
sensitivity using a source--target segment that touches
the box's closed boundary at an obstacle edge.  Shifting the target in one
coordinate by $-1/256$ makes the segment narrowly clear, while a shift of
$1/256$ makes it enter the box.  Because grid sampling places the endpoints at
cell centres at each resolution, G5 asks whether these prescribed near-contact
classifications survive sampling rather than measuring a full-field refinement
rate.  Figure~\ref{fig:validation-scene-g3} compares the G3 scalar
field at $N=64$ and its $\theta=1/2$ classification with continuous visibility.
Subsection~\ref{sec:geometry-refinement-results} then measures the G3 transition
width and boundary displacement as the grid is refined.

\begin{figure*}[t]
	\centering
	\includegraphics[width=\textwidth]{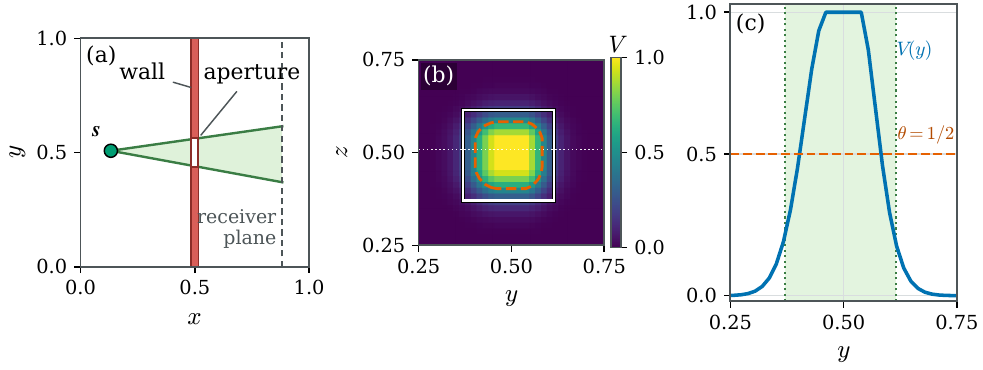}
	\caption{Narrow-aperture case G3 at $N=64$.  (a)~A central slice shows the
		grid source $\boldsymbol{s}=(17/128,65/128,65/128)$, wall aperture,
		receiver plane, and continuously visible region.
		(b)~Four-parent scalar visibility $V$ on the receiver plane at
		$x=(56+\tfrac12)/64=113/128$.
		The white outline marks the continuously visible region, the orange
		dashed contour marks $V=1/2$, and the dotted line identifies the receiver
		line at $z=(32+\tfrac12)/64=65/128$.  (c)~The blue curve gives $V$ along
		the receiver line, while the green band marks the continuously visible
		interval.  Along the receiver line, the sampled positions satisfying
		$V(y)\geq 1/2$ form the predicted visible interval.  Its endpoints are the
		first and last such positions.}
	\label{fig:validation-scene-g3}
\end{figure*}

\subsection{Accuracy under grid refinement}
\label{sec:geometry-refinement-results}

As $N\to\infty$ ($h\to0$),
Theorem~\ref{thm:nongrazing-refinement-consistency} gives eventual agreement with
continuous visibility, for any fixed threshold
$\theta\in(0,1)$, whenever a fixed source--target segment either remains
separated from every obstacle or crosses an obstacle interior.  The
experiments below instead measure how classification changes
as the same physical scenes are sampled on progressively finer grids.  Free
space supplies the base case: because no
obstacle interrupts a source--target segment in G0, every target should remain
visible.  All four methods satisfy this invariant at every tested resolution
and threshold.  With this base case established, G1--G4 test refinement in the
presence of obstacle boundaries.  At $\theta=1/2$, the false-negative rate
measured against the continuous-visibility label decreases at each refinement
for every method.  Because this rate also counts targets for which the
thresholded prediction agrees with the blocked-cell label rather than the
continuous-visibility label, Table~\ref{tab:geometry-refinement} reports the
residual false-negative rate.

\begin{table*}[t]
	\centering
	\small
	\caption{Residual false-negative rate at $\theta=1/2$, in percent.  Each
		entry is the total residual count over G1--G4 divided by the total number of
		continuously visible free targets in those scenes.}
	\label{tab:geometry-refinement}
	\begin{tabular}{@{}lrrrr@{}}
		\toprule
		Method                   & $N=32$   & $N=64$   & $N=128$  & $N=256$  \\
		\midrule
		Four-parent              & $0.0227$ & $0.0255$ & $0.0166$ & $0.0099$ \\
		Edge-gated six-neighbour & $0.0899$ & $0.0852$ & $0.0525$ & $0.0277$ \\
		Controlled tetrahedral   & $0.0043$ & $0.0036$ & $0.0008$ & $0.0005$ \\
		Edge-gated tetrahedral   & $0.5467$ & $0.3946$ & $0.2153$ & $0.1119$ \\
		\bottomrule
	\end{tabular}
\end{table*}

For all four methods, the pooled residual false-negative rate at $N=256$ is
lower than at $N=32$.  At $N=256$, the controlled tetrahedral variant gives
the lowest rate, $0.0005\%$, followed by the four-parent, edge-gated
six-neighbour, and edge-gated tetrahedral methods at $0.0099\%$, $0.0277\%$,
and $0.1119\%$, respectively.  This ordering describes the pooled G1--G4
result at $\theta=1/2$, not a scene- or threshold-independent ranking.

False positives follow a different refinement pattern.  Every method has zero
false positives at $N=32$ and $N=64$, but nonzero rates appear for some
methods at $N=128$.  At $N=256$, the pooled G1--G4 rates range from zero to
$0.0131\%$, and the largest individual scene--method rate is $0.0402\%$.
Every false-positive prediction disagrees with both comparison labels.  False
positives therefore remain rare at $\theta=1/2$, but do not decrease
monotonically as the grid is refined.  The edge-gated tetrahedral method
trades zero false positives for the largest residual false-negative rate,
while the controlled variant combines the smallest residual false-negative
rate with a negligible false-positive rate.

The residual rate counts misclassified targets but does not show how the
visibility boundary changes spatially.  The narrow-aperture scene shows the
spatial boundary behaviour directly.  At $N=64$,
Figure~\ref{fig:validation-scene-g3}(c) illustrates how the sampled four-parent
field changes gradually from values near one inside the aperture to values near
zero outside the aperture.  For each of the two aperture boundaries crossed by
the receiver line, linear interpolation between neighbouring samples locates
the positions where $V(y)=0.1$ and $V(y)=0.9$.  The distance between those
positions is the transition width.  Averaging the two transition widths gives
a contraction from $0.1132$ at $N=32$ to $0.04765$ at $N=256$.  Boundary
displacement instead compares the sampled intervals.  It is the distance
between each predicted endpoint and the corresponding endpoint of the
continuously visible interval.  At $N=256$, both endpoint displacements for
every method are at most $0.0078125=2h$, or two cell widths.  The narrower
transition and smaller displacement show that the fixed aperture is resolved
more sharply under refinement, while the remaining displacement rules out
equality with continuous visibility at finite resolution.

The threshold $\theta=1/2$ used to calculate the reported classification errors
and boundary displacements is illustrative rather than preferred.  Lowering
the threshold can reduce
false negatives at the cost of more false positives.  For the four-parent
method in G4 at $\theta=1/10$, the false-positive rates are $4.12\%$,
$7.26\%$, $7.88\%$, and $7.09\%$ at
$N\in\{32,64,128,256\}$, respectively.  The observed trade-offs therefore
support neither a preferred threshold nor a universal monotone error trend.

The final comparison concerns the three near-contact G5
queries introduced in Subsection~\ref{sec:validation-design}.  For the
prescribed narrowly clear query, cell-centre sampling shifts the segment enough
that both the continuous-visibility and blocked-cell labels report blockage
through $N=256$.  The disagreement with the originally prescribed clear
segment arises before any recurrence is evaluated and therefore cannot be
attributed to a visibility method.  G5 consequently marks the scope of
Theorem~\ref{thm:nongrazing-refinement-consistency}: it applies to fixed
segments with positive obstacle clearance or an intersection with obstacle
interior.  Boundary contact without an interior intersection remains outside
the theorem and depends on the declared sampling and boundary conventions.

\subsection{Runtime, throughput, memory, and method cost}
\label{sec:performance-results}

After comparing geometric accuracy, we turn to two questions about
computational cost.  The first compares the CPU performance of the four
visibility methods.  The second keeps the four-parent recurrence fixed and
measures how execution schedule, source position, and field transfers affect
CPU and GPU throughput.

Both studies run on a laptop with an Intel Core i9-13980HX CPU and an NVIDIA
GeForce RTX 4070 Laptop GPU.  The four-method comparison covers G0 and G4 at
$N\in\{64,128,256,512\}$.  The scheduling study combines the same two scenes
and four resolutions with four source positions in unit-cube coordinates: centre
$(1/2,1/2,1/2)$, near edge $(1/32,1/2,1/2)$, near corner
$(1/32,1/32,1/32)$, and asymmetric $(1/4,3/5,2/3)$.  As in
Section~\ref{sec:validation-design}, the source is placed at the centre of the
cell containing the prescribed position.  The resulting
$2\times4\times4=32$ scene--position--resolution cases form the schedule
comparison.  The centred and near-corner cases are also evaluated at $N=1000$.
In G4, the obstacle moves with the source so that their relative geometry
remains fixed.

All CPU measurements use eight OpenMP threads, each pinned to a distinct
physical CPU core.  All four methods use
single-precision fields, weights, and local accumulation.  In the four-method
comparison, the four methods are evaluated in a deterministically shuffled
order in each of seven timing rounds.  In every round, each method performs two
untimed warm-up solves followed by three timed solves, and its reported time is
the median of the resulting 21 measurements.  Each scheduling result is the
median of seven solves after two CPU or three CUDA warm-up solves.
Throughput is the number of returned visibility values divided by elapsed
time.  A transfer-inclusive CUDA measurement additionally counts one
transmission-field transfer to the GPU and one visibility-field transfer back
to the CPU.

Table~\ref{tab:native-method-performance} compares the four CPU
implementations on G4.

\begin{table*}[t]
	\centering
	\small
	\setlength{\tabcolsep}{3pt}
	\renewcommand{\arraystretch}{1.08}
	\caption{CPU performance of the four methods on G4 using the benchmark
		laptop.  Relative runtime is the runtime of each method divided by the
		four-parent runtime at $N=512$.}
	\label{tab:native-method-performance}
	\begin{tabular}{@{}lrrrrr@{}}
		\toprule
		                         & \multicolumn{4}{c}{Throughput (million values/s)}
		                         & \shortstack{Relative runtime                                                             \\at $N=512$} \\
		\cmidrule(lr){2-5}
		Method                   & $N=64$                                            & $N=128$ & $N=256$ & $N=512$ &        \\
		\midrule
		Four-parent              & $1036$                                            & $1490$  & $647$   & $543$   & $1.00$ \\
		Controlled tetrahedral   & $723$                                             & $887$   & $394$   & $411$   & $1.32$ \\
		Edge-gated six-neighbour & $398$                                             & $724$   & $434$   & $236$   & $2.30$ \\
		Edge-gated tetrahedral   & $395$                                             & $375$   & $184$   & $193$   & $2.81$ \\
		\bottomrule
	\end{tabular}
\end{table*}

The four-parent implementation has the highest throughput at every tested G4
resolution.  At $N=512$, the controlled tetrahedral, edge-gated six-neighbour,
and edge-gated tetrahedral implementations require $1.32$, $2.30$, and $2.81$
times its runtime, respectively.  The four-parent implementation also has the
highest throughput at every tested G0 resolution.  The edge-gated tetrahedral
method is slower than the controlled variant because it tests every active
predecessor--target edge for voxel intersection, whereas the controlled
variant applies only target-cell blocking.

The second study keeps the four-parent recurrence fixed and changes only how
each outward layer is scheduled.  The persistent-shell CPU schedule retains
one OpenMP thread team throughout the sweep, whereas the earlier CPU schedule
creates a new team for every layer.  The octant CPU schedule instead groups
cells by their source-centred octant.  On the GPU, the fused schedule processes
an entire layer in one kernel rather than launching one kernel for each owned
face.  Table~\ref{tab:native-schedule-performance} compares the baseline and
revised schedules and the effects of source position.

\begin{table*}[t]
	\centering
	\small
	\setlength{\tabcolsep}{2pt}
	\renewcommand{\arraystretch}{1.08}
	\caption{Four-parent scheduling performance on the benchmark laptop.
		The schedule-change results are geometric means over the 32
		scene--position--resolution cases defined above.
		Transfer-inclusive CUDA values count one transmission-field and one
		visibility-field transfer.}
	\label{tab:native-schedule-performance}
	\textit{Schedule comparison: geometric mean (billion values/s)}
	\par\smallskip
	\begin{tabular}{@{}lrrr@{}}
		\toprule
		\shortstack[l]{Execution scope                    \\and schedule change}
		 & Earlier schedule & Revised schedule & Gain     \\
		\midrule
		\shortstack[l]{OpenMP solve                       \\thread team per layer $\rightarrow$ persistent shell}
		 & $0.745$          & $0.797$          & $6.9\%$  \\
		\shortstack[l]{CUDA solve                         \\face kernels $\rightarrow$ fused layer}
		 & $2.40$           & $4.28$           & $78.6\%$ \\
		\shortstack[l]{CUDA with transfers                \\face kernels $\rightarrow$ fused layer}
		 & $0.829$          & $0.996$          & $20.2\%$ \\
		\bottomrule
	\end{tabular}

	\medskip
	\textit{Source-position comparison at $N=1000$ (billion values/s)}
	\par\smallskip
	\begin{tabular}{@{}lrrrr@{}}
		\toprule
		Field and source
		                & \shortstack{Persistent-shell                               \\CPU}
		                & \shortstack{Octant                                         \\CPU}
		                & \shortstack{CUDA                                           \\solve}
		                & \shortstack{CUDA with                                      \\transfers} \\
		\midrule
		G0, centre      & $0.542$                      & $1.657$ & $2.284$ & $0.906$ \\
		G4, centre      & $0.536$                      & $1.632$ & $2.284$ & $0.914$ \\
		G0, near corner & $0.613$                      & $0.286$ & $2.332$ & $0.921$ \\
		G4, near corner & $0.427$                      & $0.182$ & $2.334$ & $0.922$ \\
		\bottomrule
	\end{tabular}
\end{table*}

Retaining one CPU thread team throughout the sweep raises geometric-mean
throughput by $6.9\%$.  Fusing each CUDA layer raises solve throughput by
$78.6\%$ and transfer-inclusive throughput by $20.2\%$.

The $N=1000$ results show why source position also matters.  For a centred
source, octant OpenMP outperforms both persistent-shell OpenMP and
transfer-inclusive CUDA.  Near a corner, one octant contains most of the cells,
so octant OpenMP falls below persistent-shell OpenMP, while fused CUDA remains
faster than both CPU schedules even after transfers are counted.  CPU--GPU
ordering therefore depends on source position and on whether field transfers
are included.

The $N=1000$ cases also expose the memory limit.  The transmission and
visibility fields require $8.00$~GB ($7.45$~GiB) in single precision and
nearly exhaust the laptop GPU's available memory.  The CPU validation process
peaks at $11.18$~GiB because it retains a second output field for comparison.
The second output field is not required by the full-field algorithm.  As
established in Section~\ref{sec:implementations}, the algorithm requires storage
linear in the number of grid cells.

Taken together, the geometry results show lower residual false-negative rates
at $N=256$ than at $N=32$ and sharper aperture boundaries under refinement,
but do not establish exact agreement with continuous visibility at any tested
resolution, a monotone error trend for every threshold, or a preferred
threshold.  The performance results show that the four-parent
implementation achieves the highest throughput at every tested resolution in
the G4 CPU comparison, while method rules, execution schedule, source position,
and field transfers all affect computational cost.  Combined with the
finite-path analysis and nonbinary
consequence in Sections~\ref{sec:exact-2d-path-model}--
\ref{sec:multidimensional-extension}, the geometry and performance evidence
complete the paper's theory--algorithm--implementation chain.

\section{Scope, Limitations, and Conclusion}
\label{sec:conclusion}

This paper establishes that the dominant-axis visibility recurrence is a
finite-path model of scalar transmission rather than merely a local
interpolation rule.
For every source and target, the recurrence computes the initial visibility
value $V_s$ times the uniform average of the transmission products over a finite
family of canonical paths.  The endpoint-conditioned hypergeometric bridge gives the
probability that a uniformly selected path passes through each cell on an
intermediate 2D layer.  In 3D, the two transverse-coordinate bridge laws form
the 3D product bridge, whose target-entry probabilities are the four bilinear
weights.  Assigning one transverse-advance schedule to every transverse
coordinate yields the multilinear recurrence in any fixed spatial dimension.
The finite-path model therefore identifies both what
the scalar visibility value measures and why the local predecessor weights have
their stated form.

The finite-path average defines the value that
Algorithm~\ref{alg:causal-full-field} must return at each cell and that the
NumPy, OpenMP, and CUDA implementations must reproduce up to floating-point
roundoff.  Exhaustive
enumeration checks the recurrence independently on cases with tractable path
counts, while the three implementations agree numerically on the returned fields.
Together, exhaustive path enumeration and full-volume numerical agreement
establish that the implementations preserve the dominant-axis recurrence.
Section~\ref{sec:validation} then compares the four methods' thresholded
visibility classifications with continuous visibility on successively refined
grids and separately reports runtime, throughput, and memory use on the
benchmark laptop.

The recurrence is exact for the declared canonical path family, but finite-path
exactness does not imply equality with continuous visibility at a fixed grid
resolution.  A fixed source--target query is nongrazing here when its segment
either stays a positive distance from every obstacle or enters an obstacle
interior.  For any such query and any fixed threshold
strictly between zero and one,
Theorem~\ref{thm:nongrazing-refinement-consistency} proves that the thresholded
recurrence classification agrees
with continuous visibility on every sufficiently fine grid, provided that every
obstacle-intersecting cell is blocked, the grid source and target cells remain
free, and no additional decay is applied.  Boundary contact without an
interior intersection remains dependent on the sampling and boundary
conventions.  The grid-refinement results likewise support neither a preferred
threshold nor a convergence rate uniform over all queries.

Two further limits concern probability and transmission.  The finite-path
model defines probability through uniform selection within the canonical path
family, not through uncertainty in the map.  When the transmission field
instead assigns each cell its individual probability of being free under a
random-map model, multiplying those probabilities along a path represents
expected path transmission only when the probability that all cells on the
path are free equals the product of their individual free probabilities, as
stated in Section~\ref{sec:recurrence}.  For deterministic
transmissions, averaging
complete path products accounts for cross-layer joint incidence, but the
prescribed attenuation remains a scalar model and does not represent
scattering, refraction, wavelength, or multiple-bounce transport.

The paper's contribution is the connected theory--algorithm--\hspace{0pt}implementation
chain, not an isolated interpolation stencil or causal sweep.  The
source-relative derivation determines the local visibility update, the
path-average identities prove that each scalar value equals the finite-path
average for the corresponding source-to-cell canonical path family, and
Algorithm~\ref{alg:causal-full-field} evaluates that quantity throughout the
grid.  The dominant-axis visibility field is therefore more than an efficient
collection of fractional outputs: the finite-path average gives every cell
value a precise mathematical meaning that is preserved by the full-field
algorithm and its NumPy, OpenMP, and CUDA implementations.

\section*{Acknowledgments}

The authors thank Renaud Detry, Ömür Arslan, and Patrick Vandewalle for
questions arising from Ibrahim Ibrahim's preliminary doctoral defence that
prompted him to revisit his earlier notes on path interpretations of
two-dimensional visibility and develop the probabilistic and multidimensional
formulation presented in this article.

\section*{Declarations}

\paragraph{Funding}
Part of this research was carried out within the framework of Flanders Make's
SBO project ARENA (Agile \& REliable NAvigation).

\paragraph{Competing interests}
The authors have no relevant financial or non-financial interests to disclose.

\paragraph{Data availability}
No external datasets were used. The numerical results reported in this article
can be regenerated using the implementations and scripts identified below.

\paragraph{Code availability}
The implementations and scripts used to reproduce the numerical results are
publicly available at\\
\href{https://github.com/IbrahimSquared/hypergeometric-grid-visibility}
{\nolinkurl{https://github.com/IbrahimSquared/}}\\
\href{https://github.com/IbrahimSquared/hypergeometric-grid-visibility}
{\nolinkurl{hypergeometric-grid-visibility}}.

\paragraph{Author contributions}
Ibrahim Ibrahim developed the mathematical formulation and software, conducted
the validation and performance studies, prepared the figures, and wrote the
original draft. Wilm Decré supervised the research and critically reviewed the
methodology, results, and manuscript. Both authors contributed to the study
conception and approved the final manuscript.

\paragraph{Use of generative AI}
OpenAI Codex was used as an interactive assistant for software development,
Python scripts for numerical analysis and data visualization, checks of
mathematical derivations, theorem statements, and proofs, and manuscript
development and consistency review. The authors reviewed and verified all
resulting mathematical statements, software, numerical results, visualizations,
and text and take full responsibility for the content of this article.

\bibliographystyle{elsarticle-num}
\bibliography{references}

\end{document}